\documentclass[11pt]{article}

\usepackage[english]{babel}
\usepackage{amsmath}
\usepackage{amssymb}
\usepackage{amsthm}
\usepackage{graphicx}
\usepackage{bbm, xcolor}

\usepackage{cite}

\usepackage{hyperref}
\hypersetup{
    colorlinks=true,
    linkcolor=blue,
    filecolor=magenta,
    urlcolor=cyan,
    citecolor=blue
}

\newtheorem{theorem}{Theorem}[section]

\newtheorem{proposition}[theorem]{Proposition}

\newtheorem{corollary}[theorem]{Corollary}
\newtheorem{lemma}[theorem]{Lemma}

\theoremstyle{definition}

\newtheorem{example}[theorem]{Example}

\newtheorem{remark}[theorem]{Remark}
\newtheorem{examples}[theorem]{Examples}

\newcommand{\ep}{\epsilon}
\newcommand{\R}{\mathbb R}
\newcommand{\BS}{\mathrm {BS}}
\newcommand{\E}{\mathbb E}

\newcommand{\p}{\mathbb P}
\newcommand{\J}{\mathcal J}
\newcommand{\I}{\mathcal I}

\title{Short Maturity Forward Start Asian Options in Local Volatility Models}

\author{
Dan Pirjol
\thanks{Email: \texttt{dpirjol@gmail.com}}  \ , \
Jing Wang
\thanks{Department of Mathematics, University of Illinois Urbana--Champaign. Email: \texttt{wangjing@illinois.edu}} \ , \
Lingjiong Zhu
\thanks{Department of Mathematics, Florida State University. Email: \texttt{zhu@math.fsu.edu}} \ 
}

\date{}

\begin{document}

\maketitle

\begin{abstract}
We study the short maturity asymptotics for prices of forward start Asian 
options under the assumption that the underlying asset follows a local 
volatility model. We obtain asymptotics for the cases of out-of-the-money, 
in-the-money, and at-the-money, considering both fixed strike and 
floating Asian options. The exponential decay of the price of an 
out-of-the-money forward start Asian option is handled using 
large deviations theory, and is controlled by a rate function which
is given by a double-layer optimization problem. 
In the Black-Scholes model, the calculation of the rate function is simplified 
further to the solution of a non-linear equation. We obtain closed form for the 
rate function, as well as its 
asymptotic behaviors when the strike is extremely large, small, or close 
to the initial price of the underlying asset.  \\ 


\end{abstract}

\tableofcontents

\section{Introduction}

Asian options are among the most popular traded instruments in the equity and
commodity markets. A great variety of numerical and exact methods have been
proposed for their pricing 
\cite{Curran,Du,DufresneReview,FMW,FPP2013,GY,RogersShi,Linetsky,Vecer},
see Boyle and Potapchik \cite{BoylePot} for a survey.
Most of these methods are numerically
and computationally intensive. Monte Carlo methods require a long time, and
the calculation of the Greeks is delicate. Methods based on inverting Laplace
transforms by numerical integration \cite{GY,FMW} require special attention 
in the small maturity and/or small volatility region. 

Recently the pricing of Asian options with continuous time averaging
has been studied in the short maturity asymptotic regime \cite{PZAsian,PZAsianCEV,MLP}
in the local volatility model. This approach uses large deviations theory,
and relates the short maturity asymptotics to a rate function for the time average
of the asset price. Explicit results for the rate function can be obtained
for the Black-Scholes model and the CEV model. This approach avoids the numerical
issues in the short maturity and/or volatility region noted for the methods
mentioned above. A related approach considers 
the asymptotics of the Asian options with discrete time averaging in the limit
of a very large number of averaging dates has been proposed in \cite{PZdiscrete}
under the Black-Scholes model. 

In this paper, we are interested in the forward start Asian options
in the short maturity asymptotic regime, assuming that the asset price follows
the local volatility model. 
A forward start option becomes active at a specified date in the future; 
however, the premium is paid in advance,
and the time to maturity and the underlying asset are established at the time 
the forward start option is purchased, see e.g. \cite{MR}.
Forward starting European options have been studied in various works
in the literature. Such options have the payoff $(S_{T_2} - k S_{T_1})^+$
where $t < T_1 < T_2$ with $t$ the pricing date, and $k$ is the strike. 
More exotic derivatives such as cliquets depend also on the joint
distribution of the asset price at future times, see \cite{Gatheral} for
a survey. 

The forward start options depend on the future level of volatility, and several
approaches have been proposed to describe this quantity and model its 
dynamics. Stochastic volatility models are a popular approach. For the
purpose of pricing volatility derivatives a convenient approach is the
the variance curve model \cite{Buehler}. 
Dynamical models for implied volatility \cite{Schon} have been also 
proposed, although an arbitrage-free dynamical specification leads to 
complicated consistency conditions \cite{RogersTeh}. 
The paper \cite{GlassermanWu} introduced different notions of forward 
volatilities for forecasting purposes. 
Empirical studies of the forward smile have
been carried out in a series of papers  by Bergomi \cite{Bergomi}. 
The paper \cite{EM} empirically studies the forward smile in Sato models and
runs comparisons with a suite of models including Heston and local volatility 
models for forward smile sensitive products such as cliquets.
More closely related to the approach followed here, the asymptotics of the 
forward start European options
has been studied in various settings of small- and large-maturity in
\cite{JackRoome,JackRoome2} under the assumption that the 
asset price follows an exponential L\'{e}vy and Heston model,

Asian options are defined with an averaging period $[T_1,T_2]$ with $T_2>T_1$.
The total averaging period is $T_2-T_1$ and  the option pays at time $T_2$.
For example, a forward start (fixed strike) Asian call option with strike $K$ 
pays $\left(\frac{1}{T_{2}-T_{1}}\int_{T_{1}}^{T_{2}}S_{t}dt - K\right)^+$ 
at time $T_2$, hence the 
price of this option is given by an expectation under risk-neutral measure
\begin{equation}\label{eq-intro-C}
C(T_1,T_2)=e^{-rT_{2}}\mathbb{E}\left[\left(\frac{1}{T_{2}-T_{1}}\int_{T_{1}}^{T_{2}}S_{t}dt-K\right)^{+}\right].
\end{equation}
Similarly, the price of the forward start put fixed strike Asian option is 
given by
\begin{equation}
P(T_1,T_2)=e^{-rT_{2}}\mathbb{E}\left[\left(K-\frac{1}{T_{2}-T_{1}}\int_{T_{1}}^{T_{2}}S_{t}dt\right)^{+}\right].
\end{equation}
Another popular instrument are floating strike Asian options. The
forward start floating strike Asian call option has the price
\begin{equation}\label{eq-intro-Cf}
C_f(T_1,T_2)=e^{-rT_{2}}
\mathbb{E}\left[\left(\kappa S_{T_2} - 
\frac{1}{T_{2}-T_{1}}\int_{T_{1}}^{T_{2}}S_{t}dt\right)^{+}\right]\,,
\end{equation}
and the price of the forward start put floating strike Asian option is given by
\begin{equation}
P_f(T_1,T_2)=e^{-rT_{2}}
\mathbb{E}\left[\left(\frac{1}{T_{2}-T_{1}}\int_{T_{1}}^{T_{2}}S_{t}dt
- \kappa S_{T_2}\right)^{+}\right].
\end{equation}

The papers \cite{PZAsian,PZAsianCEV,MLP} assume that the averaging period 
starts at the valuation time $(T_1=0)$. However, in practice the averaging 
period of the Asian options may start also at some time $T_1 > 0$ in the future. 
Recall that the total averaging period is $T_2-T_1$ and the option pays at 
time $T_2$.
Forward start Asian options have been considered in the literature,
and analytical approximations have been proposed in Bouaziz et al. \cite{BBC} and Tsao et al. \cite{TCL}, also 
including quanto effects, see Chang et al. \cite{quanto}. Vanmaele 
et al. \cite{VDL} considered forward starting Asian options with discrete time 
averaging under the Black-Scholes
model and derived upper bounds on the prices of these instruments 
from comonotonicity.

In this paper, we are interested in the limiting behavior of forward start
Asian options when $T_1$ and $T_2$
approach to $0$ with a constant ratio. 
Let $T_{1}=\tau T$ and $T_{2}=T$  for some fixed ratio $\tau\in(0,1)$ and 
maturity $T>0$. Clearly when the call option is out-of-the-money ($S_0<K$) or 
at-the-money ($S_0=K$), we have $C(T):=C(T_1, T_2)\to 0$ as $T\to 0$.
Similarly, when the put option is out-of-the-money ($S_{0}>K$)
or at-the-money ($S_{0}=K$), we have $P(T):=P(T_1, T_2)\to 0$ as $T\to 0$.
However, the limiting behaviors are quite different. In fact the out-of-the-money case is 
governed by rare events which shall be captured by large deviation techniques, 
whereas the at-the-money case is governed by the fluctuations about the typical events. 

When $\tau=0$, this falls into the case of a standard Asian option, 
whose short maturity asymptotics are studied in \cite{PZAsian}. 
In this paper, we 
consider the strict forward case, namely when $\tau>0$. This situation requires
special consideration. 

We assume that the 
underlying asset follows a local volatility model. From \eqref{eq-intro-C} one 
can realize that the short maturity pricing problem is equivalent to estimating
the probability of the average asset price exceeding the strike price, i.e. 
$\p\left(\frac{1}{1-\tau}\int_{\tau}^{1}S_{Tt}dt>K \right)$ , which is a rare 
event when $K>S_0$ and $T\to0$. A natural approach is to use 
large deviation theory and the contraction principle \cite{DZ}.  
For instance, we can obtain that the price of an forward start Asian call ($S_0$, $K$, 
$\tau$, $T$) satisfies, when $K>S_0$,
\begin{equation*}
\lim_{T\to 0}T\log C(T)=-
\inf_{
\substack{\frac{1}{1-\tau}\int_{\tau}^{1}e^{g(t)}dt=K
\\
g(0)=\log S_{0}, g\in\mathcal{AC}[0,1]}}\frac{1}{2}\int_{0}^{1}\left(\frac{g'(t)}{\sigma(e^{g(t)})}\right)^{2}dt=:-{\cal I}_{\rm fwd}(S_0,K,\tau).
\end{equation*}
This suggests that the major contribution to the probability  $\p\left(\frac{1}{1-\tau}\int_{\tau}^{1}S_{Tt}dt>K \right)$ is closely related to the minimum energy of absolutely continuous paths started from $S_0$ whose arithmetic average during $(\tau,1)$ is $K$.  
This minimum energy is known as the rate function of a large deviation problem. Although it is rather complicated to look for a closed form for the rate function, we are able to reduce it to a double-layer variational problem
\begin{equation}\label{eq-intro-I}
{\cal I}_{\rm fwd}(S_0,K,\tau)= \inf_{c\in \mathbb{R}}
\left\{ \frac12 c^2 \tau + \frac{1}{1-\tau}
{\cal I}\left(S_0 e^{F^{-1}(c\tau)}, K \right) \right\},
\end{equation}
where $F(\cdot)$ is given by
$
F(\cdot) = \int_0^{\cdot} \frac{dz}{\sigma(S_0 e^z)}\, 
$
and ${\cal I}(x,K)$ is given by the variational problem
\begin{equation*}
{\cal I}(x,K)  = \inf_{\substack{\varphi\in \mathcal{AC}[0,1], \varphi(0)=0 \\ \int_0^1  e^{\varphi(u)}du = {K}/{x}}}\left\{ \frac12
\int_0^1  \left(\frac{\varphi'(u)}{\sigma(x e^{\varphi(u)})}\right)^2du\right\}.
\end{equation*}
We find out that the optimal path corresponding to the minimal energy is 
indeed a patch of the optimal path of a European option from time $0$ to $\tau$ 
and the optimal path of an Asian option from time $\tau$ to $1$. 
It also coincides with the intuition that a forward start Asian option is an 
interpolation of a European option and a standard Asian option. 

Furthermore, under the assumption of the Black-Scholes model (where volatility 
of the underlying asset is a constant $\sigma$), we are able to compute the
 rate function explicitly and obtain for $K>S_{0}$,
\begin{equation*}
\lim_{T\to0}T\log C(T)=-\frac{1}{2\sigma^2(1-\tau)^2}\left({\tau}\beta^2\tanh^2\frac{\beta}{2}+(1-\tau)\beta^2-2(1-\tau)\beta\tanh\frac{\beta}{2} \right),
\end{equation*}
and for $K<S_{0}$,
\begin{equation*}
\lim_{T\to0}T\log P(T)=-\frac{2}{\sigma^2(1-\tau)^2}\left({\tau}\xi^2\tan^2\xi-(1-\tau)\xi^2+(1-\tau)\xi\tan\xi \right),
\end{equation*}
where $\beta\in(0,\infty)$  and $\xi\in(0,\pi/2)$ are the unique 
solutions of 
\begin{align*}
\frac{\sinh\beta}{\beta}=\frac{K}{S_0}e^{-\frac{\tau}{1-\tau}\beta\tanh\frac{\beta}{2}},
\quad
\frac{\sin (2\xi)}{2\xi}=\frac{K}{S_0}e^{\frac{2\tau}{1-\tau}\xi\tan{\xi}}.
\end{align*}
We also obtain the optimal path explicitly by gluing the optimal path of an 
European option ($S_0$, $S_0e^{c\sigma\tau}$) over the period $[0, \tau)$
with the optimal path of a standard Asian option $(S_0e^{c\sigma\tau}, K)$ 
over the averaging period $[\tau, 1]$. This is given explicitly as
\begin{equation*}
f(t)=\begin{cases}
ct & 0\le t\le\tau
\\
c\tau+\varphi\left(\frac{t-\tau}{1-\tau}\right) & \tau<t\le 1
\end{cases}\,,
\end{equation*}
where $\varphi(\cdot)$ is the optimal path of a standard Asian option under 
Black-Scholes model with maturity $1$, and $c$ is uniquely determined by the 
$C^1$ smoothness of $f$. 
This also holds for general local 
volatility models, namely the unique $\arg\min$ $c$ of \eqref{eq-intro-I} is 
such that the optimal path has continuous first derivative,
which is guaranteed by \eqref{eq-S-t-condition}.

The paper is organized as follows. 
In Section \ref{sec-loc-vol}, we study
the asymptotics of forward start Asian options in the out-of-the-money, 
at-the-money and in-the-money regimes, under the local volatility model.  
In Section \ref{sec-BS-fwd}, we focus 
on the special case of the Black-Scholes model, for which explicit results can 
be obtained. We also study the asymptotic expansions of the rate functions 
in the cases of deep out-of-the-money and around at-the-money, and introduce 
the notions of $\tau$-AATM and $\tau$-DOTM that are more suitable for forward 
starting Asian options. At last in 
Section \ref{sec-float} we discuss the asymptotics for short maturity forward
start Asian options with floating strike, considering both cases of newly issued
and seasoned floating strike Asian option. These cases correspond to the valuation
period being prior or during the averaging period, respectively. We also
present numerical tests for the asymptotic formulas, by comparing with the
analytical approximation of \cite{TCL} and Monte Carlo simulation.
An Appendix summarizes the notations used
for the various rate functions giving the short maturity asymptotics.

\section{Local volatility model}\label{sec-loc-vol}
In this section we study the short maturity asymptotics for the price of a 
forward start Asian option, under the assumption that the underlying asset 
follows the local volatility model, see e.g. \cite{Gatheral} for an overview. 
We assume the underlying asset price is given by the solution of the 
following stochastic differential equation
\begin{equation}\label{eq-S-t}
dS_t=(r-q)S_tdt+\sigma(S_t)S_t\,dW_t,\quad S_0>0,
\end{equation}
where $W_t$ is a standard  Brownian motion, $r\ge0$ is the risk-free rate,
$q\geq 0$ is the continuous dividend yield, and $\sigma(\cdot)$ is the local
volatility. We impose the following assumptions on the local volatility
$\sigma(\cdot )$: there exist $0 < \underline{\sigma} <  \overline{\sigma} <
\infty$ and $M,\alpha>0$ such that
\begin{align}\label{eq-S-t-condition}
& 0< \underline{\sigma}\le \sigma(x)\le \overline{\sigma}<\infty,\quad \forall \ x\in[0,\infty),\\
& |\sigma(e^x)-\sigma(e^y)|\le M|x-y|^\alpha,\quad \forall\ x,y\in \mathbb{R}\nonumber.
\end{align}
Under these assumptions it was shown in the paper of Varadhan 
\cite{Varadhan67} that the log asset price paths, and hence the asset price paths satisfy 
large deviations principles. This property can be shown to hold in a wider 
class of 
local volatility models including the CEV model \cite{DRYZ,BaldiCaram},
which do not satisfy the conditions \eqref{eq-S-t-condition}.
For simplicity
we restrict ourselves here to the class of functions $\sigma(\cdot )$ 
with property \eqref{eq-S-t-condition}, with the understanding that the 
results can be generalized appropriately.

We are interested in the asymptotic behavior of the forward start Asian call 
option with parameters $(S_0,K,T,\tau)$
\begin{equation*}
C(T)=e^{-rT}\mathbb{E}\left[\left(\frac{1}{(1-\tau)T}\int_{\tau T}^{T}S_{t}dt-K\right)^{+}\right]\,,
\end{equation*}
and the forward start Asian put option $(S_0,K,T,\tau)$
\begin{equation*}
P(T)=e^{-rT}\mathbb{E}\left[\left(K-\frac{1}{(1-\tau)T}\int_{\tau T}^{T}S_{t}dt\right)^{+}\right]\,,
\end{equation*}
as $T\to0$. Obviously $C(T)$ and $P(T)$ converge to $0$ in the out-of-the-money and at-the-money cases. However, the causes are very different, which then lead to very different vanishing speed. When it is out-of-the-money, the rapid decrease to $0$ of the option prices comes from the extremely small probability of positive payoff when maturity $T\to0$. The decrease is exponentially fast. When it is at-the-money, the option prices drop to $0$ comes from the drop of Gaussian fluctuation as $T\to0$. The speed is at the scale of $\sqrt{T}$.

\subsection{Out-of-the-money and in-the-money asymptotics}
First we study the asymptotic for out-of-the-money  case. Recall the definitions for a forward start Asian option. Let $A(T,\tau)$ be the forward averaged asset price under the risk-neutral measure
\begin{equation}\label{eq-A-T}
A(\tau T,\tau):=\frac{1}{(1-\tau)T}\int_{\tau T}^T\E[S_t]dt,
\end{equation}
then 
\begin{equation*}
A(\tau T,\tau)=\begin{cases}
\frac{S_0\left( e^{(r-q)T}-e^{(r-q)\tau T}\right)}{(1-\tau)(r-q)T}
&\text{when $r-q\not=0$},\\
S_0 
&\text{when $r-q=0$}.
\end{cases}
\end{equation*}  
A forward start Asian call option is said to be {out-of-the-money} if 
$K>A(\tau T,\tau)$ and {in-the-money} if $K<A(\tau T,\tau)$. However, note that
$A(\tau T,\tau)=S_0+O(T)$ when $T\to0$. This implies that in the small maturity 
regime, $K>A(\tau T,\tau)$ is  equivalent to $K>S_0$. Hence throughout the 
rest of this paper, we say a forward start Asian call option is 
{out-of-the-money} if $K>S_0$ and {in-the-money} if $K<S_0$. Correspondingly 
we call a forward start Asian put option {out-of-the-money} if $K<S_0$ and 
{in-the-money} if $K>S_0$.

First we prove the following lemma that allows us to transfer the price 
estimate problem to probabilistic estimates of rare events.
\begin{lemma}\label{lemma-TCT}
Let $C(T)$, $P(T)$, $T>0$, $\tau\in (0,1)$ be as given above. Then we have
\begin{equation}\label{eq-TCT}
\lim_{T\rightarrow 0}T\log C(T)
=\lim_{T\rightarrow 0}T\log\mathbb{P}\left(\frac{1}{(1-\tau)T}\int_{\tau T}^{T}S_{t}dt\geq K\right),
\end{equation}
and 
\begin{equation}\label{eq-TPT}
\lim_{T\rightarrow 0}T\log P(T)
=\lim_{T\rightarrow 0}T\log\mathbb{P}\left(\frac{1}{(1-\tau)T}\int_{\tau T}^{T}S_{t}dt\le K\right).
\end{equation}
\end{lemma}
\begin{proof}
 The proof of \eqref{eq-TPT} is similar to \eqref{eq-TCT}. We only prove \eqref{eq-TCT} here. The idea follows similarly as in \cite{PZAsian}.  
First by H\"older's inequality,
for any $p, p'>1$, $\frac1p+\frac{1}{p'}=1$ we have
\begin{align*}
C(T)&\le e^{-rT}\mathbb{E}\left[\bigg|\frac{1}{(1-\tau)T}\int_{\tau T}^{T}S_{t}dt-K\bigg|
\mathbbm{1}_{\frac{1}{(1-\tau)T}\int_{\tau T}^{T}S_{t}dt\ge K}\right]\\
&\le e^{-rT}\left(\mathbb{E}\left[\bigg|\frac{1}{(1-\tau)T}\int_{\tau T}^{T}S_{t}dt-K\bigg|^p\right]
\right)^{\frac1p}\p\left(\frac{1}{(1-\tau)T}\int_{\tau T}^{T}S_{t}dt\ge K\right)^{\frac{1}{p'}}.
\end{align*}
We claim that when $T\to 0$, there exists a constant $C>0$ such that
\begin{equation}\label{eq-claim}
\log \mathbb{E}\left[\bigg|\frac{1}{(1-\tau)T}\int_{\tau T}^{T}S_{t}dt-K\bigg|^p\right]\le C.
\end{equation}
Assume the above estimate, we can  easily observe for any $1<p'<2$ that
\begin{equation*}
\limsup_{T\rightarrow 0}T\log C(T)\le \frac{1}{p'}\log  \p\left(\frac{1}{(1-\tau)T}\int_{\tau T}^{T}S_{t}dt\ge K\right).
\end{equation*}
By letting $p'\to 1$ we then obtain  the upper bound. On the other hand, note for any $\epsilon>0$,
\begin{align*}
C(T)&\ge e^{-rT}\mathbb{E}\left[\bigg(\frac{1}{(1-\tau)T}\int_{\tau T}^{T}S_{t}dt-K\bigg)\mathbbm{1}_{\frac{1}{(1-\tau)T}\int_{\tau T}^{T}S_{t}dt\ge K+\ep}\right]\\
&\ge e^{-rT} \ep\, \p\left(\frac{1}{(1-\tau)T}\int_{\tau T}^{T}S_{t}dt\ge K+\ep \right).
\end{align*}
By letting $\ep\to 0$ we then obtain $ \liminf_{T\rightarrow 0}T\log C(T)\ge \log  \p\left(\frac{1}{(1-\tau)T}\int_{\tau T}^{T}S_{t}dt\ge K\right)$, 
which implies \eqref{eq-TCT}. Now we are let to prove \eqref{eq-claim}.
When $p>2$, by convexity of $x\to x^p$ on $(0,\infty)$ we have that
\begin{align*}
\mathbb{E}\left[\bigg|\frac{1}{(1-\tau)T}\int_{\tau T}^{T}S_{t}dt-K\bigg|^p\right]\le 2^{p-1}\left( \mathbb{E}\left[\bigg(\frac{1}{(1-\tau)T}\int_{\tau T}^{T}S_{t}dt\bigg)^p\right]+K^p\right).
\end{align*}
Moreover,
\begin{equation*}
\mathbb{E}\left[\bigg(\frac{1}{(1-\tau)T}\int_{\tau T}^{T}S_{t}dt\bigg)^p\right]\le\frac{1}{(1-\tau)T}\int_{\tau T}^{T} \mathbb{E}(S_{t}^p)dt,
\end{equation*}
and $\mathbb{E}(S_{t}^p)$ solves the differential equation
\begin{equation*}
d\mathbb{E}(S_{t}^p)=\left(p(r-q)\mathbb{E}(S_{t}^p)+\frac12 p(p-1)\mathbb{E}[S_{t}^p\sigma(S_t)^2]\right)dt.
\end{equation*}
Using \eqref{eq-S-t-condition} we obtain that $\mathbb{E}(S_{t}^p)\le S_0^pe^{(p(r-q)+\frac12 p(p-1)\overline{\sigma}^2)t}$, hence
\begin{equation*}
\frac{1}{(1-\tau)T}\int_{\tau T}^{T} \mathbb{E}(S_{t}^p)dt\le S_0^pe^{(p(r-q)+\frac12 p(p-1)\overline{\sigma}^2)T}.
\end{equation*}
This implies \eqref{eq-claim}. We then complete the proof.
\end{proof}

From Lemma \ref{lemma-TCT} we can reduce the logarithmic estimates for prices of out-of-the-money forward start Asian options to the probability estimate of the rare event that the underlying asset price goes up from $S_0$ to $K$ within time $T\to0$. A classic tool is 
the large deviation theory. For the definition and basic properties of the large deviation theory,
we refer to \cite{DZ}.

\begin{theorem}\label{thm-otm}
Assume the asset price $S_t$ follows the local volatility model as in \eqref{eq-S-t} and \eqref{eq-S-t-condition}. Then the price of an out-of-the-money forward start Asian call option ($K>S_0$) satisfies
\begin{equation}\label{eq-out-call}
\lim_{T\to0}T\log C(T)=-{\cal I}_{\rm fwd}(S_0,K,\tau);
\end{equation}
and the price of the corresponding out-of-the-money forward start Asian put option ($K<S_0$) satisfies
\begin{equation}\label{eq-out-put}
\lim_{T\to0}T\log P(T)=-{\cal I}_{\rm fwd}(S_0,K,\tau),
\end{equation}
where
\begin{equation*}
{\cal I}_{\rm fwd}(S_0,K,\tau_1/\tau_2)=
\inf_{
\substack{\frac{1}{1-\tau}\int_{\tau}^{1}e^{g(t)}dt=K
\\
g(0)=\log S_{0}, g\in\mathcal{AC}[0,1]}}\frac{1}{2}\int_{0}^{1}\left(\frac{g'(t)}{\sigma(e^{g(t)})}\right)^{2}dt.
\end{equation*}
\end{theorem}

\begin{proof}
From Lemma \ref{lemma-TCT} we just need to have a large deviation estimate 
for $\frac{1}{(1-\tau)T}\int_{\tau T}^{T}S_{t}dt$. 
The idea is to use the contraction principle from large deviations theory 
(see e.g. \cite{DZ}).
Let $X_{t}=\log S_{t}$.  First note 
\begin{equation*}
\frac{1}{(1-\tau)T}\int_{\tau T}^{T}S_{t}dt
=\frac{T}{(1-\tau)T}\int_{\tau}^{1}S_{tT}dt
=\frac{1}{1-\tau}\int_{\tau}^{1}e^{X_{tT}}dt.
\end{equation*}
On the other hand we know that $\mathbb{P}(X_{tT}\in\cdot, t\in[0,1])$ 
satisfies a sample path large deviation principle on $L_{\infty}([0,1],\R)$ 
with the rate function (see e.g. \cite{Varadhan67})
\begin{equation*}
I(g)=\begin{cases}
\frac{1}{2}\int_{0}^{1}\left(\frac{g'(t)}{\sigma(e^{g(t)})}\right)^{2}dt,
&\text{for $g(0)=\log S_{0}, g\in\mathcal{AC}[0,1]$}
\\
+\infty
&\text{otherwise}
\end{cases},
\end{equation*}
where $\mathcal{AC}[0,1]$ is the space of absolutely continuous functions.

Since the map $g\mapsto\frac{1}{1-\tau}\int_{\tau}^{1}e^{g(x)}dx$
from $L_{\infty}([0,1],\R)$ to $\mathbb{R}^{+}$ is a continuous map, 
by contraction principle (see e.g. \cite{DZ}), we have
\begin{align}
\lim_{T\rightarrow 0}T\log C(T)
&=\lim_{T\rightarrow 0}T\log\mathbb{P}\left(\frac{1}{(1-\tau)T}\int_{\tau T}^{T}S_{t}dt\geq K\right)
\\
&=-\inf_{
\substack{\frac{1}{1-\tau}\int_{\tau}^{1}e^{g(t)}dt=K
\\
g(0)=\log S_{0}, g\in\mathcal{AC}[0,1]}}\frac{1}{2}\int_{0}^{1}\left(\frac{g'(t)}{\sigma(e^{g(t)})}\right)^{2}dt.
\nonumber
\end{align}
Hence we obtain \eqref{eq-out-call}.
Following the same argument we can easily obtain \eqref{eq-out-put}. 
\end{proof}
As a by-product, we can also obtain the estimate for an in-the-money forward start Asian option, by using put-call parity. We present it in the following corollary.
\begin{corollary}
Assume the asset price $S_t$ satisfies \eqref{eq-S-t} and \eqref{eq-S-t-condition}. Then the price of an in-the-money 
forward start Asian call option ($K<S_0$) satisfies
\begin{equation}\label{eq-in-call}
C(T)=S_0-K+\left(\frac12(r-q)(1+\tau)-(S_0-K)r\right)T+O(T^2),\qquad \mbox{as $T\to0$},
\end{equation}
and the price of the corresponding in-the-money forward start Asian put option ($K>S_0$) satisfies
\begin{equation}\label{eq-in-put}
P(T)=K-S_0-\left(\frac12(r-q)(1+\tau)+(S_0-K)r\right)T+O(T^2), \qquad \mbox{as $T\to0$}.
\end{equation}
\end{corollary}
\begin{proof}
From the put-call parity we have
\begin{align*}
C(T)-P(T)
&=e^{-rT}\E\left[\frac{1}{(1-\tau)T}\int_{\tau T}^TS_tdt-K \right]
\\
&=\begin{cases}
e^{-rT}\left(\frac{S_0\left( e^{(r-q)T}-e^{(r-q)\tau T}\right)}{(1-\tau)(r-q)T}-K\right)
&\text{if $r-q\not=0$}\\
e^{-rT}(S_0-K)
&\text{if $r-q=0$}
\end{cases}
\\
&=\begin{cases}
S_0-K+\left(\frac12(r-q)(1+\tau)-(S_0-K)r\right)T+O(T^2) &\text{if $r-q\not=0$}
\\
(S_0-K)(1-rT)+O(T^2) &\text{if $r-q=0$}
\end{cases},
\end{align*}
as $T\to0$. 
On the other hand, when $K<S_0$, from \eqref{eq-out-put} we know that $P(T)\ll O(T^2)$\footnote{For $a, b>0$, we say $a\gg b$ if $\frac{a}{b}\to+\infty$, and $a\ll b$ if $\frac{a}{b}\to 0+$.} when $T\to0$, hence 
\begin{equation*}
C(T)=S_0-K+\left(\frac12(r-q)(1+\tau)-(S_0-K)r\right)T+O(T^2).
\end{equation*} 
We can obtain \eqref{eq-in-put} using the same argument.
\end{proof}

\subsection{At-the-money asymptotics}
In this section we consider at-the-money case, namely when $S_{0}=K$. First we notice that it is much more likely for the asset price to hit $K=S_0$ after an extremely short time  comparing to out-of-the-money case. Though the probability of this event is still small, it is at the scale of Gaussian fluctuation, 
namely $\sqrt{T}$, which is significantly larger than $e^{-\I/T}$ in the 
out-of-the-money case. In the following theorem, we look for the exact 
asymptotic of an at-the-money 
forward start Asian option when $T\to0$. 

\begin{theorem}\label{thm-atm}
Consider an at-the-money forward start Asian option $(S_0=K, T, \tau)$ under 
the local volatility model \eqref{eq-S-t} and \eqref{eq-S-t-condition}. 
If in addition we assume the volatility satisfies Lipschitz continuity conditions: there exist $\alpha,\beta>0$ such that for all $x, y\ge0$,
\begin{equation*}
|\sigma(x)x-\sigma(y)y|\le\alpha |x-y|,\quad |\sigma(x)-\sigma(y)|\le\beta|x-y|.
\end{equation*}
Then the prices of the forward start Asian call and put options satisfy
\begin{equation*}
\lim_{T\rightarrow 0}\frac{1}{\sqrt{T}}C(T)
=\lim_{T\rightarrow 0}\frac{1}{\sqrt{T}}P(T)
=\sigma(S_{0})S_{0}\sqrt{\frac{1+2\tau}{6\pi}}.
\end{equation*}
\end{theorem}
\begin{proof}
The idea is to look for a linear estimate of the underlying asset $S_t$, 
which is enough to capture the small maturity asymptotic of $C(T)$ and $P(T)$ 
up to order $\sqrt{T}$. The proof is similar as in \cite{PZAsian} 
(see Theorem 6). Here we only sketch the major steps. \\
\textbf{Step 1}: We first consider an approximate process $X_t:=e^{-(r-q)t}S_t$, which is a martingale and satisfies the SDE
\begin{equation*}
dX_t=\sigma(X_te^{(r-q)t})X_tdW_t,\quad X_0=S_0.
\end{equation*}  
For any underlying process $x_t\in C([0,1],\R_+)$, we denote 
\begin{equation*}
C\left(T, x_t \right):=e^{-rT}\mathbb{E}\left[\left(\frac{1}{(1-\tau)T}\int_{\tau T}^{T}x_{t}dt-S_0\right)^{+}\right].
\end{equation*}
Then we can easily see that
\begin{align*}
e^{rT}|C\left(T, S_t \right)-C\left(T, X_t \right)|&\le \mathbb{E}\left[\frac{1}{(1-\tau)T}\int_{\tau T}^{T}\left|e^{(r-q)t}-1\right|X_tdt\right]\\
&=S_0\left|\frac{e^{(r-q)T}-e^{(r-q)\tau T}}{(r-q)(1-\tau)T}-1\right|=O(T).
\end{align*}
Hence we have
\begin{equation*}
|C\left(T, S_t \right)-C\left(T, X_t \right)|=O(T).
\end{equation*}
\textbf{Step 2}: Next we consider a further approximation of $X_t$ by a Gaussian process $\hat{X}_{t}=S_{0}+\sigma(S_{0})S_{0}W_{t}$, where $W_{t}$ is a standard
Brownian motion. Note 
\begin{equation*}
e^{rT}\left|C\left(T, X_t \right)-C\left(T, \hat{X}_t \right)\right|\le \E\left[ \max_{\tau T\le t \le T}|X_t-\hat{X}_t|\right].
\end{equation*}
We claim that $\E\left[ \max_{\tau T\le t \le T}|X_t-\hat{X}_t|\right]=O(T)$. This is due to Doob's martingale inequality and the fact that
\begin{align}\label{eq-X-t-2}
\E\left[ \max_{\tau T\le t \le T}|X_t-\hat{X}_t|\right]\le 2\left(\E[(X_T-\hat{X}_T)^2]\right)^{\frac12}
\le 2\sqrt{M}T
\end{align}
for some  constant $M>0$. The detailed proof of the last inequality in \eqref{eq-X-t-2} can be found in \cite{PZAsian} 
(see page 21-22). 
At the end, combining Step 1 and Step 2, we obtain that
\begin{equation*}
C(T)=C\left(T, \hat{X}_t\right)+O(T).
\end{equation*}
At the end we just need to compute $C\left(T, \hat{X}_t\right)$ for the Gaussian process $\hat{X}_t$. Note
\begin{equation*}
\frac{1}{(1-\tau)T}\int_{\tau T}^{T}\hat{X}_{t}dt-S_{0}
=\sigma(S_{0})S_{0}\frac{1}{(1-\tau)T}\int_{\tau T}^{T}W_{t}dt,
\end{equation*}
which is a Gaussian random variable with mean $0$ and variance
\begin{align*}
&\sigma(S_{0})^{2}S_{0}^{2}\frac{1}{((1-\tau)T)^{2}}
\mathbb{E}\left[\left(\int_{\tau T}^{T}W_{t}dt\right)^{2}\right]
=\frac{1}{3}\sigma(S_{0})^{2}S_{0}^{2}(1+2\tau )T.
\nonumber
\end{align*}
Hence we have
\begin{align*}
\mathbb{E}\left[\left(\frac{1}{(1-\tau)T}\int_{\tau T}^{T}\hat{X}_{t}dt-S_{0}\right)^{+}\right]
&=\frac{1}{\sqrt{3}}\sigma(S_{0})S_{0}\sqrt{1+2\tau}\sqrt{T}\,\mathbb{E}[Z\mathbbm{1}_{Z>0}]
\\
&=\sigma(S_{0})S_{0}\sqrt{\frac{(1+2\tau)T}{6\pi}},
\end{align*}
where $Z$ is a standard Gaussian random variable with mean zero and variance 
one. The result for the put option can be proved similarly.
\end{proof}

\subsection{Discussions on variational problem}
In this section we further discuss the variational problem for out-of-the-money case given in Theorem \ref{thm-otm}. 
We want to analyze the rate function
\begin{eqnarray}\label{1LV}
{\cal I}_{\rm fwd}(S_0,K,\tau) = \inf_{f\in\mathcal{A}(K/S_0,\tau)} \frac12
\int_0^1\left(\frac{f'(t)}{\sigma(S_0 e^{f(t)})}\right)^2 dt,
\end{eqnarray}
where
\begin{eqnarray}\label{2LV}
\mathcal{A}(K/S_0,\tau)=\bigg\{f\in \mathcal{AC}[0,1]\bigg|  f(0)=0,  \frac{1}{1-\tau} \int_\tau^1  e^{f(t)} dt= \frac{K}{S_0}\bigg\}.
\end{eqnarray}

We have the following result.

\begin{proposition}\label{prop1}
The solution of the variational problem (\ref{1LV}) is given by the
extremum problem
\begin{eqnarray}\label{3LV}
\mathcal{I}_{\rm fwd}(S_0,K,\tau) = \inf_{c\in \mathbb{R}}
\left\{ \frac12 c^2 \tau + \frac{1}{1-\tau}
{\cal I}\left(S_0 e^{F^{-1}(c\tau)}, K \right) \right\},
\end{eqnarray}
where ${\cal I}(x,K)$ is given by the variational problem
\begin{equation}\label{eq-cal-B-0}
{\cal I}(x,K)  = \inf_{\varphi\in \mathcal{A}(K/x,0)}\left\{ \frac12
\int_0^1  \left(\frac{\varphi'(u)}{\sigma(S_0 e^{\varphi(u)})}\right)^2du\right\}\,,
\end{equation}
and $F(\cdot)$ is defined by
\begin{equation}\label{eq-F}
F(\cdot) = \int_0^{\cdot} \frac{dz}{\sigma(S_0 e^z)}\, .
\end{equation}
\end{proposition}

\begin{proof}
The variational problem
(\ref{1LV}) with the constraint (\ref{2LV}) can be transformed into an 
unconstrained variational problem by introducing a Lagrange multiplier $\lambda$
for the functional
\begin{eqnarray}\label{Ldef}
\Lambda[f] := \frac12  \int_0^1  \left( \frac{f'(t)}{\sigma(S_0 e^{f(t)})} 
\right)^2 dt -
\lambda \left( \int_{\tau}^1 e^{f(t)} dt- \frac{K}{S_0}(1-\tau) \right)\,.
\end{eqnarray}
We split the variational problem
(\ref{1LV}) with the constraint (\ref{2LV}) into two parts and analyze them 
separately. \\
{\textbf{Part 1:}} When $0\leq t \leq \tau$, the Euler-Lagrange equation reads
\begin{eqnarray*}
\frac{d}{dt} \left( \frac{f'(t)}{\sigma(S_0 e^{f(t)})} \right) = 0\,,\qquad
0 \leq t \leq \tau\,.
\end{eqnarray*}
This implies that
$\frac{f'(t)}{\sigma(S_0 e^{f(t)})} = c$, from which we obtain that $f$ is monotone.
By integration we have that
\begin{eqnarray*}
\int_0^{f(t)} \frac{dz}{\sigma(S_0 e^z)} = F(f(t)) = ct.
\end{eqnarray*}
Clearly $F$ is a strictly increasing function. By inverting $F$ we then obtain the argmin of the rate function ${\cal I}_{\rm fwd}(S_0,K,\tau)$ in the period $[0,\tau]$:
\begin{eqnarray}\label{fdef}
f(t) = F^{-1}(ct)\, ,\quad 0\le t\le\tau.
\end{eqnarray}
The corresponding energy is 
\begin{equation*}
 \frac12 \int_0^\tau  \left( \frac{f'(t)}{\sigma(S_0 e^{f(t)})}\right)^2dt= \frac12 c^2\tau.
\end{equation*}
{\textbf{Part 2:}} Now consider the region $\tau\leq t \leq 1$. 
We want to find $f(t)\in \mathcal{AC}[\tau,1]$, such that 
$f(\tau )=F^{-1}(c\tau)$, $\frac{1}{(1-\tau)}\int_{\tau }^{1}S_0e^{f(t)}dt=K$, 
and has the minimum energy $\frac12 \int_\tau^1  \left( \frac{f'(t)}{\sigma(S_0 e^{f(t)})}\right)^2dt$.
Let us reparametrize $f(t)$, $t\in[\tau,1]$ as follows. 
Consider the variable $u\in[0,1]$ such that for all $t\in[\tau, 1]$ we have
\begin{eqnarray*}
t = \tau + u(1-\tau).
\end{eqnarray*}
We also center the function $f$ by letting $\varphi: [0,1]\to \mathbb{R}$ be 
such that
\begin{eqnarray}\label{eq-varphi}
 F^{-1}(c\tau) + \varphi(u) =f(t)=f(\tau + u(1-\tau)).
\end{eqnarray}
Then clearly we have
\begin{eqnarray*}
dt = (1-\tau) du\,,\qquad f'(t) = \frac{1}{1-\tau} \varphi'(u)\,.
\end{eqnarray*}
The function $\varphi(u)$ satisfies the boundary condition $\varphi(0)=0$ and \eqref{2LV} can be written as
\begin{eqnarray}\label{8LV}
\bigg\{\varphi\in \mathcal{AC}[0,1]\bigg|  \varphi(0)=0, \int_0^1 e^{\varphi(u)} du  = \frac{K}{S_0} e^{-F^{-1}(c\tau)}\bigg\}.
\end{eqnarray}
From \eqref{1LV} we know for any $c\in\mathbb{R}$,
\begin{eqnarray*}
{\cal I}_{\rm fwd}(S_0,K,\tau) \le
\frac12 c^2 \tau + \frac{1}{(1-\tau)}
\inf_{\varphi\in \mathcal{A}( K/S_0,\tau)}\bigg\{\frac12\int_0^1  \left(\frac{\varphi'(u)}
{\sigma(S_0 e^{F^{-1}(c\tau)} e^{\varphi(u)})}\right)^2du\bigg\}.
\end{eqnarray*}
Let the contribution from $F^{-1}(c\tau)$ be
absorbed into a redefinition of $S_0$, then the right hand side of the above inequality is indeed
\begin{equation*}
\frac12 c^2 \tau+\frac{1}{(1-\tau)} \mathcal{I}(S_0e^{F^{-1}(c\tau)},K):=G_{c}
\end{equation*}
Hence we obtain
\begin{equation*}
{\cal I}_{\rm fwd}(S_0,K,\tau) \le \inf_{c\in\mathbb{R}}G_c.
\end{equation*}
On the other hand, let  $G_{0}=\inf_{c\in\mathbb{R}}G_c$. Then for any $\ep>0$, there exists a $c'\in\mathbb{R}$ such that
\begin{equation*}
G_{0}>G_{c'}-\ep\ge{\cal I}_{\rm fwd}(S_0,K,\tau)-\ep.
\end{equation*}
Let $\ep\to 0$ we obtain $G_{0}\ge {\cal I}_{\rm fwd}(S_0,K,\tau)$.
This yields (\ref{3LV}) and  concludes the statement.
\end{proof}

\begin{remark} 
We know that when $S_0=K$ (at-the-money), the option price is of scale $\sqrt{T}$ as $T\to0$, hence the rate function ${\cal I}_{\rm fwd}(S_0,K,\tau)$ vanishes
for $K=S_0$. This can be easily checked by noting that $c=0$ and $\varphi=0$ solves the minimization problem.
\end{remark}

We call a path $f(t)\in \mathcal{AC}[0,1]$ the \textbf{optimal path} 
of ${\mathcal I}_{\rm fwd}(S_0,K,\tau)$ if it is the minimizer of the 
variational problem (\ref{1LV}) with the constraint (\ref{2LV}).

Next, we study  the continuity of the derivative of the optimal path $f(t)$.  
From the above proposition we know that $f(t)$ is a $C^1$ function on both 
$(0,\tau)$ and $(\tau,1)$. In the language of variational calculus, 
$\tau$ is the so-called corner point. In fact we know that $f'(t)$ is continuous 
at $t=\tau$ as well, hence $f\in C^1([0,1],\R)$. 
This is guaranteed by Erdmann-Weierstrass condition 
(see for example \S12.6 in Chapter~4 of \cite{CH} (p.~260) or 
\cite{Bliss} (p.~203)), which we present as below.

\begin{theorem}[Erdmann-Weierstrass corner conditions]
If a function $y(t)$ is an extremum of the variational problem
\begin{equation}
\frac{\delta}{\delta y(t)} \int_{a}^b L(y,y') dt =0\,,
\end{equation}
then $\partial_{y'} L$ and $y' \partial_{y'} L - L$ must be continuous at each
corner point $\tau$ of $y(t)$, that is points where $y'(t)$ may have different values
on each side of $\tau$.
\end{theorem}

In our case the function $L(y,y')$ of the variational problem (\ref{Ldef}) is,
up to an unessential constant, 
\begin{equation}\label{eq-L-f-f'}
L(f,f') = \frac12 \left( \frac{f'(t)}{\sigma(S_0 e^{f(t)})}\right)^2 - 
\lambda \mathbbm{1}_{\{\tau \leq t \leq 1\}}  e^{f(t)}\,,
\end{equation}
such that 
\begin{equation}
\partial_{f'} L(f,f') = \frac{f'(t)}{\sigma^2(S_0 e^{f(t)})}\,.
\end{equation}
Since $\sigma(x)$ is continuous everywhere, it follows by the 
Erdmann-Weierstrass
condition for $\partial_{f'} L(f,f')$ that $f'(t)$ is continuous at the
corner point $\tau$. This proves that the optimal solution $f(t)$ of this
variational problem is $C^{1}$ on $(0,1)$.

\section{Black-Scholes model and approximation}\label{sec-BS-fwd}
In this section we consider forward start Asian options under Black-Scholes model, which is a special case of the local volatility model with $\sigma(\cdot)$ in \eqref{eq-S-t} being a constant function. 
The variational problem significantly simplifies in this case and we are 
able to obtain a closed form for the rate function 
${\cal I}_{\rm fwd}(S_0,K,\tau)$.

\subsection{Rate functions under the Black-Scholes model}

Let $\sigma(\cdot)=\sigma>0$. The result of Proposition~\ref{prop1} simplifies 
as follows.
\begin{proposition}\label{prop-BS-rate}
Assume the underlying asset price $S_t$ follows Black-Scholes model. 
Then the price of an out-of-the-money forward start 
Asian call option ($K>S_0$) satisfies
\begin{equation}\label{eq-out-call-BS}
\lim_{T\to0}T\log C(T)=-\frac{1}{\sigma^2} 
\mathcal{J}_{\rm fwd}^{\rm (BS)}(K/S_0,\tau),
\end{equation}
and the price of an out-of-the-money forward start Asian put option 
($K<S_0$) satisfies
\begin{equation}\label{eq-out-put-BS}
\lim_{T\to0}T\log P(T)=-\frac{1}{\sigma^2} 
\mathcal{J}_{\rm fwd}^{\rm (BS)}(K/S_0,\tau),
\end{equation}
where
\begin{equation}\label{eq-c}
\mathcal{J}_{\rm fwd}^{\rm (BS)}(K/S_0,\tau) = \inf_{c\in\R}
\left\{ \frac12 c^2 \tau + \frac{1}{1-\tau} 
\mathcal{J}_{\rm BS}\left(\frac{K}{S_0} e^{-c\tau} \right) \right\}.
\end{equation}
Here ${\cal J}_{\rm BS}(\cdot)$ is the rate function for an Asian option 
with averaging
starting at time zero, which is given by the variational problem
\begin{equation}\label{eq-J-x}
{\cal J}_{\rm BS}(x)  = \inf_{\substack{\int_0^1  e^{\varphi(u)} du= x \\
 \varphi\in {\cal AC}{[0,1]}, \varphi(0)=0
}} \bigg\{ \frac12
\int_0^1  (\varphi'(u))^2 du\bigg\}.
\end{equation}
\end{proposition}
\begin{proof}
This easily follows from Proposition \ref{prop1}. 
We just need to show when $\sigma(\cdot)=\sigma$, we have
\begin{equation*}
{\mathcal I}_{\rm fwd}(S_0,K,\tau)=\frac{1}{\sigma^2} 
{\mathcal J}_{\rm fwd}^{\rm BS}(K/S_0,\tau).
\end{equation*}
This is straightforward by noting 
$\mathcal{I}_{\rm BS}(\cdot\  ,K)=
\frac{1}{\sigma^2}\mathcal{J}_{\rm BS}(K/\, \cdot)$ and $F(x)=x/\sigma$.
\end{proof}

\begin{remark}
An alternative way to look at the forward start Asian call option with 
parameters $(S_0,K, T, \tau)$ under the Black-Scholes model is as follows. 
The underlying asset price is given by
$S_t = S_0 e^{\sigma W_t + (r-q-\frac12 \sigma^2)t}$. Then the averaging 
over $[\tau T, T]$ is given by $S_{\tau T}A_{[\tau T,T]}$, where
\begin{equation*}
A_{[\tau T,T]}= \frac{1}{(1-\tau)T}  \int_{0}^{(1-\tau)T} e^{\sigma W_t + (r-q-\frac12 \sigma^2)t} dt.
\end{equation*}
The underlying of the forward start Asian is therefore a product of two uncorrelated random variables: a log-normally distributed variable $S_{\tau T}:=X$ and $A_{[\tau T,T]}$.
The call option price can be written as
\begin{align}\label{weightedAv}
C(T) &= e^{-r T} \mathbb{E}[(XA_{[\tau T,T]} - K)^+] =
e^{-r T} \mathbb{E}\left[ X\, \mathbb{E}\left[\left(A_{[\tau T,T]} - \frac{K}{X}\right)^+\bigg|X\right]\right] \\
&=
e^{-r T} \int_{-\infty}^\infty e^{-\frac{x^2}{2\tau T}}
e^{\sigma x + (r-q-\frac12\sigma^2) \tau T}\, C_A(K /S_{\tau T}, (1-\tau)T)\frac{dx}{\sqrt{2\pi \tau T}}
\nonumber \\
&=
 \frac{e^{-r T+(r-q)\tau T}}{\sqrt{2\pi \tau T}}  \int_{-\infty}^\infty
e^{-\frac{1}{2\tau T}(x-\sigma \tau T)^2}
 C_A(K /S_{\tau T}, {(1-\tau)T})\,dx\,,
\nonumber
\end{align}
where $C_A(K/S_{\tau T},{(1-\tau)T})$ is the undiscounted Asian call
option price, with averaging starting at time zero, underlying starting price at $1$, and strike price at $K/S_{\tau T}$.
In conclusion, the price of a forward start Asian option is a Gaussian-weighted average
over the prices of Asian options starting at time zero.

It is already known that (for instance see \cite{PZAsian}) when $T\to0$, 
we have
\begin{equation*}
C_A(K/S_{\tau T},{(1-\tau)T})=\exp\left(- \frac{1}{T(1-\tau)} {\cal J}_{BS}(K/S_{\tau T})+o(1/T)\right).
\end{equation*}
We plug this asymptotic result back into \eqref{weightedAv}. 
Using the Laplace method we know that $C(T)$ is dominated by the exponential 
term
\begin{equation*}
\exp {\left(-\frac{x^2}{2\tau T}-\frac{1}{T(1-\tau)} 
\mathcal{J}_{\rm BS}\left(Ke^{-\sigma x-(r-q-\frac12\sigma^2)\tau T}/S_{0}
\right)\right)}\,,
\end{equation*}
where $\mathcal{J}_{\rm BS}(\cdot)$ is the rate function for the Asian
option in the Black-Scholes model.
Hence we have
\begin{equation*}
\lim_{T\to 0}T\log C(T)=-\inf_{x\in \R}\bigg\{\frac{x^2}{2\tau}+\frac{1}{1-\tau}
\mathcal{J}_{\rm BS}\left(Ke^{-\sigma x}/S_0\right) \bigg\}.
\end{equation*}
This coincides with \eqref{eq-c} by letting $x=c\tau/\sigma$.
\end{remark}

\begin{corollary}
In the limiting case of very small averaging period,  $\tau \to 1$, 
\eqref{eq-c} reduces to the rate function 
of an European option
\begin{equation}\label{eq-limit-J}
\lim_{\tau \to 1} {\cal J}_{\rm fwd}^{(BS)}(K/S_0,\tau) = \frac12 \log^2 \left(\frac{K}{S_0}\right).
\end{equation}
When $\tau\to0$,  \eqref{eq-c} reduces to the rate 
function of a standard Asian option with averaging starting from 
time $0$, namely
\begin{equation*}
\lim_{\tau \to 0} {\cal J}_{\rm fwd}^{(BS)}(K/S_0,\tau)= 
\mathcal{J}_{\rm BS}(K/S_0).
\end{equation*}
\end{corollary}
\begin{proof}
The second conclusion is obvious. We just show \eqref{eq-limit-J}. 
When $\tau\to1$, from \eqref{eq-c} we have that
\begin{equation*}
\lim_{\tau\to1}\mathcal{J}_{\rm fwd}^{(BS)}(K/S_0,\tau) = \inf_{c\in\R}
\left\{ \frac12 c^2  + \frac{1}{1-\tau} {\cal J}_{\rm BS}
\left(\frac{K}{S_0 e^{c}} \right) \right\}.
\end{equation*}
Since $ \frac{1}{1-\tau} \to+\infty$ as $\tau\to1$, the optimal $c\in\R$ is the 
smallest real number for which 
$\mathcal{J}_{\rm BS}\left(\frac{K}{S_0 e^{c}}\right)$ vanishes. 
Moreover, from \eqref{eq-J-x} we know that ${\cal J}_{\rm BS}
\left(\frac{K}{S_0 e^{c}}\right)=0$ as soon as $S_0 e^c= K$. 
Hence we obtain the optimal $c=\log\left(\frac{K}{S_0} \right)$, which gives
\eqref{eq-limit-J}. This agrees with the intuition that when $\tau$ is close 
to $1$, almost no averaging takes place. It therefore falls back to a standard 
European option.
\end{proof}

\subsection{Discussions on variational problem and the optimal path}\label{sec-var-opt-path}
In this section we look for a closed form expression for the rate function 
$\frac{1}{\sigma^2} {\mathcal J}_{\rm fwd}^{\rm (BS)}(K/S_0,\tau)$ of an 
out-of-the-money 
forward start Asian option under the Black-Scholes model. From \eqref{eq-c} this
amounts to solving a double-layer variational problem. 
We know that when $\tau\in(0,1)$, the rate function is an interpolation 
between the extreme cases corresponding to $\tau=0$ and $\tau=1$. 
It can also be observed by looking at its optimal path. 
Recall that a path $f(t)\in \mathcal{AC}[0,1]$ is the \textbf{optimal path} 
of ${\mathcal J}_{\rm fwd}^{(BS)}(S_0,K,\tau)$ if it satisfies
\begin{equation*}
f(0)=0,\quad  \frac{1}{1-\tau}\int_\tau^1e^{f(t)}dt=\frac{K}{S_0}, \quad
{\mathcal J}_{\rm fwd}^{(BS)}(K/S_0,\tau)=\frac12\int_0^1 [f'(t)]^2dt.
\end{equation*} 
From the proof of Proposition \ref{prop1} we know that $f(t)$ is a combination 
of two continuous paths, namely the optimal paths for the European option 
($S_0$, $S_0e^{c\tau}$, $\tau$) and the optimal path for the Asian option 
($S_0e^{c\tau}$, $K$, $1-\tau$).  Explicitly we have
\begin{equation}\label{eq-f-fwd}
f(t)=\begin{cases}
ct & 0\le t\le\tau\\
c\tau+\varphi_c\left(\frac{t-\tau}{1-\tau}\right) & \tau<t\le 1
\end{cases}\,,
\end{equation}
where $\varphi_c(\cdot)\in \mathcal{AC}[0,1]$ is the argmin of 
$ {\cal J}_{\rm BS}\left(\frac{K}{S_0 e^{c}} \right)$.
Here $c$ is uniquely determined by the minimizer of \eqref{eq-c}, which shall 
be  discussed in detail later.  
The explicit forms of $ {\cal J}_{\rm BS}(x)$ and its argmin 
$\varphi(\cdot\, ,x)$ were computed in \cite{PZAsian} (see Proposition 12). 
We recall these results below.

\begin{lemma}\label{lemma-J}
Let $\mathcal{J}_{\rm BS}(x)$ be given as in \eqref{eq-J-x}. Then we have
\begin{equation*}
\mathcal{J}_{\rm BS}(x)=\begin{cases}
\frac{1}{2}\beta^2-\beta\tanh(\beta/2) & x\ge 1\\
2\xi(\tan \xi-\xi) & x<1
\end{cases}\,,
\end{equation*}
and for $u\in[0,1]$
\begin{equation*}
\varphi(u,x)
=\begin{cases}
\beta u-2\log\left(\frac{e^{\beta u}+e^\beta}{1+e^\beta} \right) & x\ge 1
\\
\log\left(\frac{\cos \xi}{\cos^2(\xi(u-1))}\right) & x<1
\end{cases}\,,
\end{equation*}
where $\beta\in[0,\infty)$ and $\xi\in [0,\pi/2)$ are the unique solutions
\footnote{For $x=1$, we define $\beta=0$ and $\xi=0$.}
of the following equations
\begin{equation*}
\frac{1}{\beta}\sinh\beta=x,\quad \frac{1}{2\xi}\sin(2\xi)=x.
\end{equation*}
\end{lemma}

\begin{figure}[ht]
\begin{center}
\includegraphics[height=45mm]{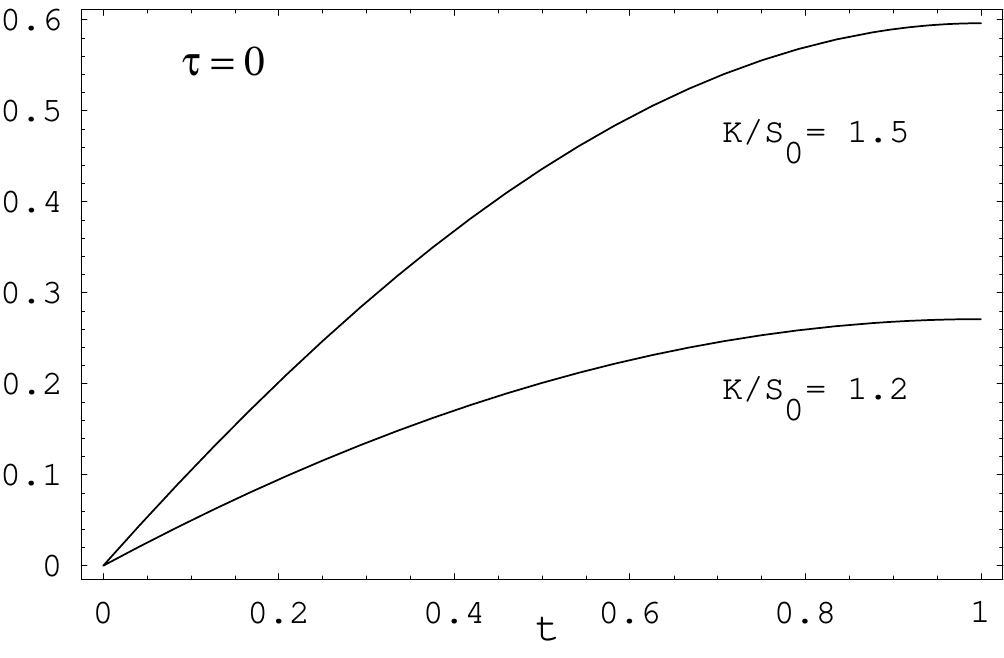}
\includegraphics[height=45mm]{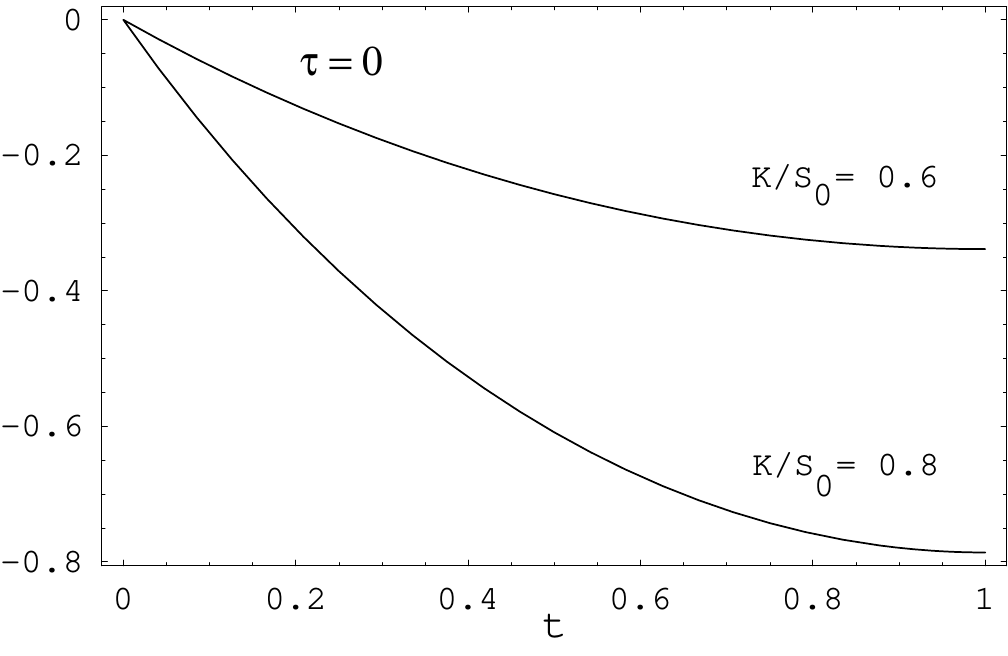}
\end{center}
\caption{Optimal paths $\varphi(\cdot\, ,x)$ for $x=K/S_0>1$ (left) and $K/S_0<1$ (right)
for an Asian option in the BS model. The averaging period is $[0,1]$,
corresponding to $\tau=0$. The
values of $K/S_0$ are as shown.}
\label{Fig:1a}
\end{figure}
The left plot in Figure~\ref{Fig:1a} shows the optimal path $\varphi(u,x)$ for 
two strikes with $x=K/S_0 > 1$ (out-of-the-money Asian  call options).
This is a concave function satisfying $\varphi(0,x)=0$ and the transversality 
condition $\varphi'(1,x)=0$. 
Recall that $\varphi(u,x)$ satisfies also $\int_0^1 e^{\varphi(u,x)} du = K/S_0=x$.\\

By Erdmann-Weierstrass corner conditions we know that $f$ is in $C^1([0,1],\R)$, which in fact, uniquely determines the constant $c$. 
In the proposition below, we given a simple and direct verification of $f\in C^1([0,1],\R)$ and obtain an explicit expression 
for $c$. 

\begin{proposition}\label{prop-f'-c}
Let ${\mathcal J}_{\rm fwd}^{(BS)}(K/S_0,\tau)$ be given as in \eqref{eq-c} and 
$f\in \mathcal{AC}[0,1]$ be the optimal path as described in \eqref{eq-f-fwd}. 
Then  $f\in C^1([0,1],\R)$.
Moreover, we have
\begin{equation*}
c=\begin{cases}
\frac{1}{1-\tau}\beta_c\tanh\left( \frac{\beta_c}{2}\right) & K\ge S_0e^{c\tau} 
\\
-\frac{2}{1-\tau}\xi_c\tan \xi_c & K\le S_0e^{c\tau} 
\end{cases}\,,
\end{equation*}
where $\beta_{c}\in[0,\infty)$ and $\xi_{c}\in[0,\pi/2)$ are the unique 
solutions
\footnote{For $K=S_{0}e^{c\tau}$, we define $\beta_{c}=\xi_{c}=0$.} of 
\begin{equation}\label{betaceq}
\frac{1}{\beta_c}\sinh\beta_c=\frac{K}{S_0e^{c\tau}},\quad 
\frac{1}{2\xi_c}\sin(2\xi_c)=\frac{K}{S_0e^{c\tau}}.
\end{equation}
\end{proposition}

\begin{proof}
To show $f\in C^1([0,1],\R)$, we just need to prove that $f'(\tau-)=f'(\tau+)$. 
By \eqref{eq-f-fwd} it is equivalent to prove that
\begin{equation}\label{eq-c-varphi-0}
c=\frac{1}{1-\tau}\varphi_c'\left(0\right).
\end{equation}
From Lemma \ref{lemma-J} and the definition of $\varphi_c$ we know that 
\begin{equation}\label{eq-varphi-prime-0}
\varphi_c'(0)
=\begin{cases}
\beta_c \tanh\left(\frac{\beta_c}{2} \right) & K\ge S_0e^{c\tau},
\\
-2\xi_c\tan(\xi_c) & K< S_0e^{c\tau}.
\end{cases}
\end{equation}
One the other hand, since $c$ is the argmin of $\J_{\rm fwd}(S_0, K, \tau)$ in \eqref{eq-c}, it is  a critical point. Hence 
\begin{equation}\label{eqc}
c=-\frac{1}{\tau(1-\tau)}\,\frac{\partial}{\partial c}\J\left(\frac{K}{S_0e^{c\tau}} \right).
\end{equation}
By Lemma \ref{lemma-J} we can easily obtain that
\begin{equation*}
\frac{\partial}{\partial c}\J_{\rm BS}\left(\frac{K}{S_0e^{c\tau}} \right)
=\begin{cases}
-\tau\beta_c\tanh\left( \frac{\beta_c}{2}\right) & K\ge S_0e^{c\tau}, 
\\
2\tau\xi_c\tan \xi_c & K\le S_0e^{c\tau}.
\end{cases}
\end{equation*}
Combine the above two equations we have
\begin{equation*}
c=\begin{cases}
\frac{1}{1-\tau}\beta_c\tanh\left( \frac{\beta_c}{2}\right) & K\ge S_0e^{c\tau},
\\
-\frac{2}{1-\tau}\xi_c\tan \xi_c & K\le S_0e^{c\tau}. 
\end{cases}
\end{equation*}
Comparing with \eqref{eq-varphi-prime-0} we immediately obtain 
\eqref{eq-c-varphi-0}. This completes the proof.
\end{proof}

\begin{figure}[ht]
\begin{center}
\includegraphics[height=60mm]{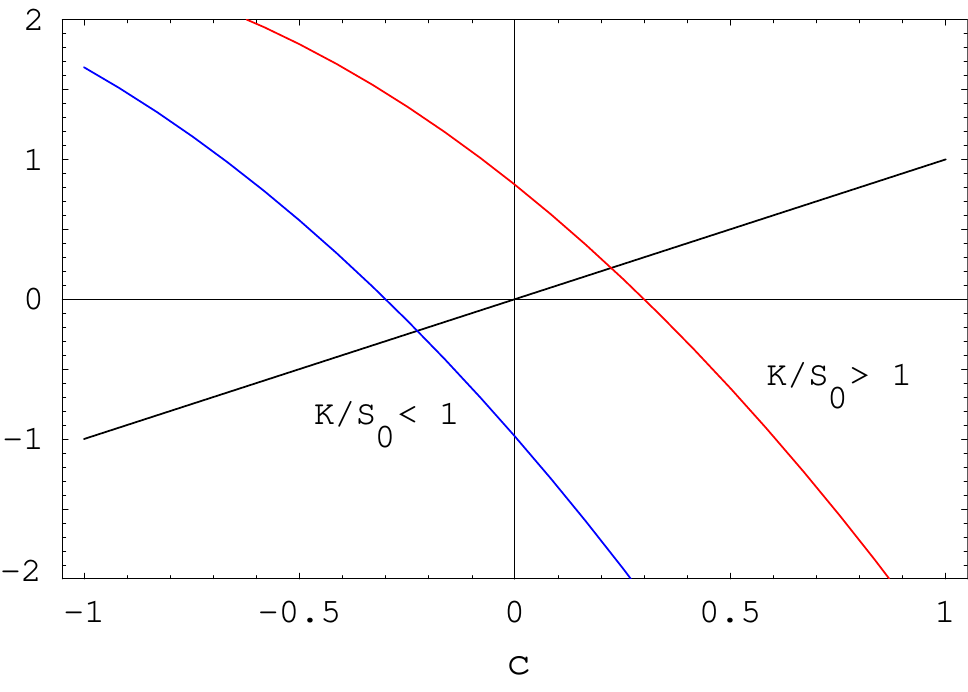}
\end{center}
\caption{
Graphical representation of the equation (\ref{eqc}) for $c$. This equation
can be represented as the intersection of the curve 
$\frac{1}{\tau}\frac{\partial}{\partial c}\J\left(\frac{K}{S_0e^{c\tau}} \right)$ and
$c (1-\tau)$ plotted vs $c$. The two cases shown correspond to $K>S_0$
(red) and $K<S_0$ (blue). The values of $c$ in the two cases are positive
and negative, respectively.}
\label{Fig:csol}
\end{figure}

From Proposition \ref{prop-f'-c} we can easily obtain the following fact.

\begin{proposition}
For a forward start Asian option with parameters $(K, S_0, \tau)$ and $\tau\in(0,1)$, let $c=S_0e^{f(\tau)}$ be as described above. Then 
\begin{itemize}
\item[(1)] If $K=S_0$, we have $c=0$;
\item[(2)] If $K>S_0$, then $c>0$ and $S_0<S_0e^{c\tau}<K$;
\item[(3)] If $K<S_0$, then $c<0$ and $S_0>S_0e^{c\tau}>K$.
\end{itemize}
\end{proposition} 

\begin{proof}
First note that the optimal path \eqref{eq-f-fwd} is piecewisely  monotone. This is because the optimal paths for a regular European option and a 
standard Asian option are both monotone (see \cite{PZAsian}, page 25). Moreover from \eqref{eq-c-varphi-0} we know that $c$ has the same sign as $\varphi'(0)$. 
This implies that the optimal path $f$ is monotone in the whole interval $[0,1]$. We can then easily conclude that $c>0$ if $K>S_0$, $c<0$ if $K<S_0$ and $c=0$ if $K=S_0$. Note $S_0e^{f(\tau)}=S_0e^{c\tau}$. By  monotonicity of $S_0e^{f(t)}$, $t\in[0,1]$ we can easily conclude that $S_0e^{c\tau }$ always falls between $S_0$ and $K$.
\end{proof}

Now we are able to derive the closed form of the rate function of an out-of-the-money forward start Asian option under Black-Scholes model.
\begin{theorem}
Let  $C(T)$ (resp. $P(T)$) be the price of an out-of-the-money forward start Asian call (resp. put) option $(S_0, K, T, \tau)$ under Black-Scholes model with volatility $\sigma$. Then we have the short maturity asymptotics of the prices as follows.
For $K>S_{0}$,
\begin{equation}\label{eq-BS-rate}
\lim_{T\to0}T\log C(T)=-\frac{1}{2\sigma^2(1-\tau)^2}\left({\tau}\beta^2\tanh^2\frac{\beta}{2}+(1-\tau)\beta^2-2(1-\tau)\beta\tanh\frac{\beta}{2} \right)\,,
\end{equation}
and for $K<S_{0}$,
\begin{equation}\label{eq-BS-rate-p}
\lim_{T\to0}T\log P(T)=-\frac{2}{\sigma^2(1-\tau)^2}\left({\tau}\xi^2\tan^2\xi-(1-\tau)\xi^2+(1-\tau)\xi\tan\xi \right),
\end{equation}
where $\beta\in(0,\infty)$  and $\xi\in(0,\frac{\pi}{2})$ are the unique solutions of 
\begin{align}\label{eq-z-beta}
\frac{\sinh\beta}{\beta}=\frac{K}{S_0}e^{-\frac{\tau}{1-\tau}\beta\tanh\frac{\beta}{2}}, \quad K> S_0\,, 
\end{align}
\begin{align}\label{eq-z-xi}
\frac{\sin (2\xi)}{2\xi}=\frac{K}{S_0}e^{\frac{2\tau}{1-\tau}\xi\tan{\xi}}, \quad K<S_0\,.
\end{align}
\end{theorem}

\begin{proof}
By combining Proposition \ref{prop-BS-rate} and Proposition \ref{prop-f'-c} 
we immediately obtain the conclusion.
\end{proof}

We show in Table~\ref{tab:1} numerical results for $c$ as given by the
solution of the equation \eqref{eqc} for a forward start Asian option
with several values of $\tau$ and different strikes $K>S_0$. 
The solution for $\beta_c$ of the
equation (\ref{betaceq}) is also shown. Substituting these numerical values
into the expressions for $f(t)$ shown above, gives the optimal paths shown
in the left plot of Figure~\ref{Fig:2a}.

These optimal paths have two pieces: a linear piece (dashed blue) corresponding
to the time interval before the averaging period, and the concave
piece (solid black) corresponding to the averaging period.

\begin{table}[!ht]
\caption{\label{tab:1} 
Numerical solution of the equation (\ref{eqc}) for $c$, for the 
forward start Asian option with $\tau$ values shown and $K>S_0$. 
For each $c$ we compute $\beta_c$
given by the solution of $\frac{1}{\beta_c}\sinh\beta_c = \frac{K}{S_0} 
e^{-c\tau}$. The last column shows the rate function of the
forward start Asian option in the Black-Scholes model.}
\begin{center}
\begin{tabular}{ccccc}
\hline
$K/S_0$ & $\tau$ & $c$ & $\beta_c$ & $\mathcal{J}_{\rm fwd}^{(BS)}(K/S_0,\tau)$\\
\hline\hline
1.0 & 0.0 & -  & 0.0     & 0.0 \\
1.1 & 0.0 & -  & 0.76340 & 0.01337 \\
1.2 & 0.0 & -  & 1.06487 & 0.04813 \\
1.3 & 0.0 & -  & 1.2873  & 0.09818 \\
1.4 & 0.0 & -  & 1.46814 & 0.15932 \\
1.5 & 0.0 & -  & 1.62213 & 0.22854 \\
\hline
1.0 & 0.25 & 0 & 0                & 0.0 \\
1.1 & 0.25 & 0.189261 & 0.5392    & 0.00904 \\
1.2 & 0.25 & 0.35968  & 0.751447  & 0.03294 \\
1.3 & 0.25 & 0.514475 & 0.907739  & 0.06794 \\
1.4 & 0.25 & 0.65612  & 1.03462   & 0.11132 \\
1.5 & 0.25 & 0.786554 & 1.14257   & 0.16109 \\
\hline
1.0 & 0.5 & 0        & 0        & 0.0 \\
1.1 & 0.5 & 0.142709 & 0.38003  & 0.00681  \\
1.2 & 0.5 & 0.272543 & 0.52806  & 0.02487 \\
1.3 & 0.5 & 0.391598 & 0.636173 & 0.05146 \\
1.4 & 0.5 & 0.501497 & 0.723307 & 0.08455 \\
1.5 & 0.5 & 0.603527 & 0.796955 & 0.12267 \\
\hline
1.0 & 0.75 & 0        & 0        & 0.0 \\
1.1 & 0.75 & 0.114339 & 0.239673 & 0.09531 \\
1.2 & 0.75 & 0.218666 & 0.332169 & 0.18232 \\
1.3 & 0.75 & 0.314588 & 0.399221 & 0.26236 \\
1.4 & 0.75 & 0.403356 & 0.452895 & 0.33647 \\
1.5 & 0.75 & 0.485962 & 0.497977 & 0.40546 \\
\hline
\end{tabular}
\end{center}
\end{table}

\begin{figure}[!ht]
\begin{center}
\includegraphics[height=40mm]{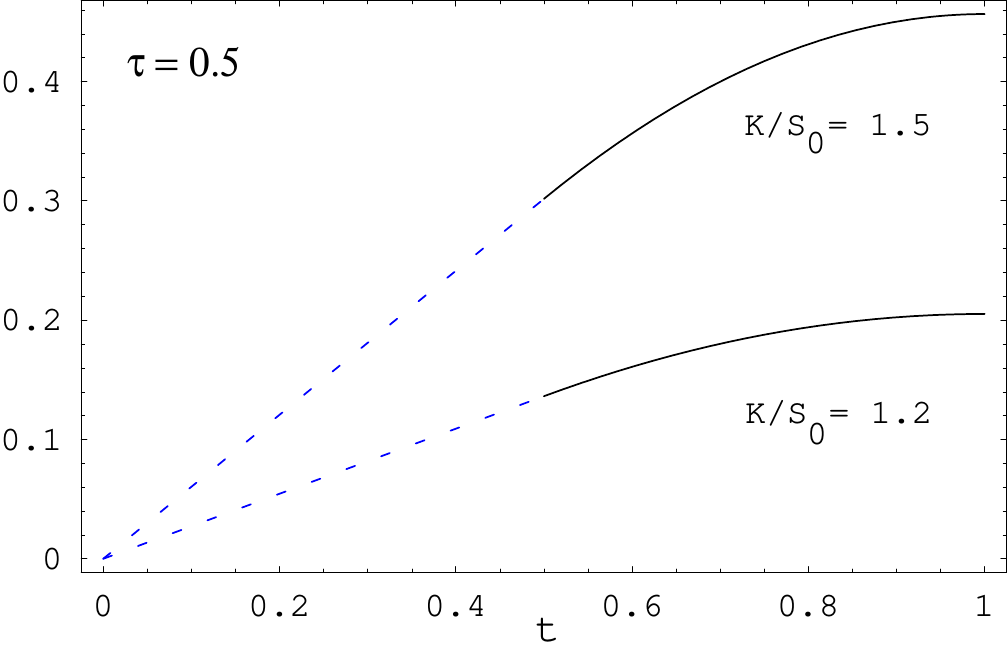}
\includegraphics[height=40mm]{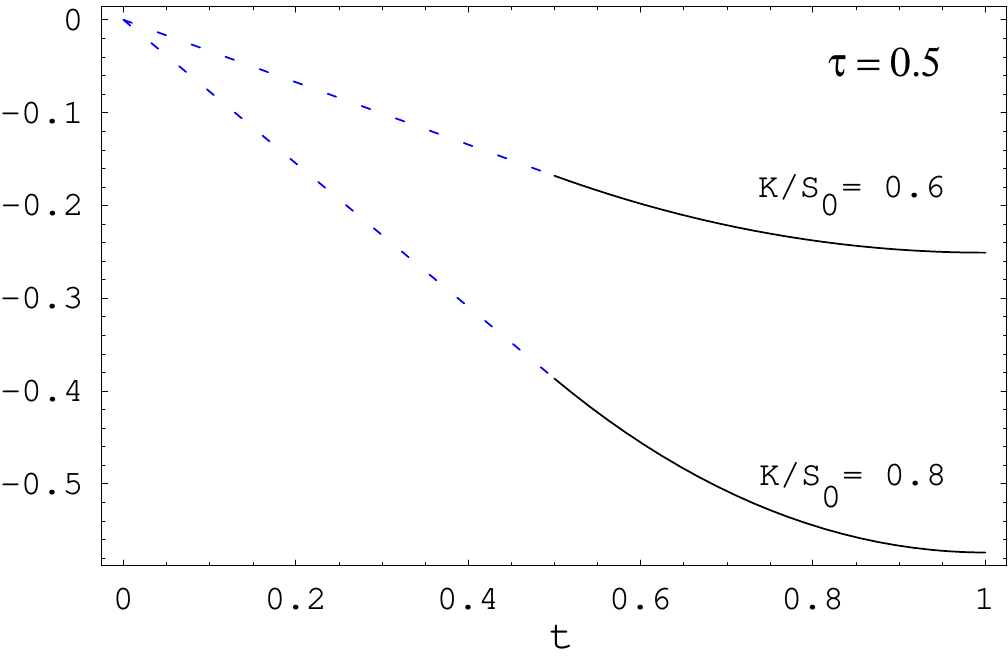}
\end{center}
\caption{Optimal path $f(t)$ given by \eqref{eq-f-fwd} for 
$K/S_0=1.2, 1.5$ (left) 
and $K/S_0=0.6, 0.8$ (right)
for a forward start Asian option in the BS model with $\tau=0.5$. The two
pieces of the path with different analytical form are shown as the 
dashed blue and solid black curves, respectively.}
\label{Fig:2a}
\end{figure}

We also show numerical values for $c$ and $\beta_c$ for 
a few
values of $K/S_0 \leq 1$ in Table~\ref{tab:2}.
The right plot in Fig.~\ref{Fig:2a} shows the optimal path $f(t)$ for 
forward start Asian options with averaging over $[\tau ,1]$ with $\tau=0.5$
and strike $K/S_0=0.8$ and 0.6. 
The optimal path is a straight line in the $[0,\tau]$ region, and 
a convex function piece in the $[\tau ,1]$ region.


\begin{table}[!ht]
\caption{\label{tab:2} 
Numerical solution of the equation (\ref{eqc}) for $c$, 
for the forward start
Asian option with $\tau$ values shown and $K<S_0$. 
For each $c$ we compute  $\xi_c$
given by the solution of $\frac{1}{2\xi_c}\sin(2\xi_c) = \frac{K}{S_0} 
e^{- c\tau}$. The last column shows the rate function of the
forward start Asian option in the Black-Scholes model.}
\begin{center}
\begin{tabular}{ccccc}
\hline
$K/S_0$ & $\tau$ & $c$ & $\xi_c$ & $\mathcal{J}_{\rm fwd}^{(BS)}(K/S_0,\tau)$ \\
\hline\hline
1.0 & 0.0 & -  & 0.0 & 0 \\
0.9 & 0.0 & -  & 0.39334 & 0.01701 \\
0.8 & 0.0 & -  & 0.56555 & 0.07823 \\
0.7 & 0.0 & -  & 0.70509 & 0.20580 \\
0.6 & 0.0 & -  & 0.83002 & 0.43735 \\
0.5 & 0.0 & -  & 0.94775 & 0.84161 \\
\hline
1.0 & 0.25 & 0         & 0        & 0 \\
0.9 & 0.25 & -0.21239  & 0.278525 & 0.01116 \\
0.8 & 0.25 & -0.453789 & 0.401176 & 0.05035 \\
0.7 & 0.25 & -0.732556 & 0.5013   & 0.12950 \\
0.6 & 0.25 & -1.06111  & 0.591885 & 0.26765 \\
0.5 & 0.25 & -1.45904  & 0.678576 & 0.49724 \\
\hline
1.0 & 0.5 & 0 & 0 & 0 \\
0.9 & 0.5 & -0.158352 & 0.197664  & 0.00834 \\
0.8 & 0.5 & -0.336107 & 0.285876  & 0.03745 \\
0.7 & 0.5 & -0.538556  & 0.358898 & 0.09584 \\
0.6 & 0.5 & -0.773475  & 0.426057 & 0.19694 \\
0.5 & 0.5 & -1.05297   & 0.491616 & 0.36342 \\
\hline
1.0 & 0.75 & 0         & 0 & 0.0 \\
0.9 & 0.75 & -0.126473 & 0.125404 & 0.00667 \\
0.8 & 0.75 & -0.267951 & 0.181998 & 0.02989 \\
0.7 & 0.75 & -0.428465 & 0.229381 & 0.07639 \\
0.6 & 0.75 & -0.613921 & 0.273526 & 0.15672 \\
0.5 & 0.75 & -0.833483 & 0.317279 & 0.28867 \\
\hline
\end{tabular}
\end{center}
\end{table}

\subsection{Further asymptotic estimates of the rate function}
Under Black-Scholes model, we have a  closed form for the logarithmic estimate for the price of an out-of-the-money 
forward start Asian option ${\cal I}_{\rm fwd}(S_0,K,\tau)$. 
It has different behavior when the option strike takes extreme values. 
In this section, we take a further step and estimate these asymptotic behaviors. 
For convenience we use log-strike $x=\log(K/S_0)$ in the rest of this section. 

Recall for a standard Asian option, we call it around ATM (AATM) if $|x|\ll1$ 
and deep OTM (DOTM) if  $|x|\gg1$. For a forward start Asian option, this 
is the case only if $\tau \sim O(1)$. To include also the limit values of 
$\tau$ we define the following regions in $(\tau, x)$ plane:

\begin{itemize}

\item $\tau$-almost-ATM ($\tau-$AATM) region. This is the region
$(1-\tau) |x| \ll 1$. It includes the AATM region.

\item $\tau$-deep-OTM ($\tau-$DOTM) region. 
This is the region $(1-\tau)|x| \gg 1$, and includes
the region of very large $x\to \infty$ and very small $-x\to \infty$ strikes.
The deep-OTM region with $x<0$ is further divided into two regions,
as will become apparent in Prop.~\ref{prop:LeftWing}.
\end{itemize}

See Figure~\ref{Fig:regions} for a graphical representation of these regions.

The relevant scale for the expansions in this section is $(1-\tau)x$ and
the form of the expansion depends on the relative size of this parameter to 1.
In order to see this we recall the relation \eqref{eq-z-beta} determining
$\beta$ for $K>S_0$ 
\[
(1-\tau)\log\left(\frac{\sinh\beta}{\beta}\right)+\tau\beta\tanh\left(\frac{\beta}{2}\right)
=(1-\tau)x\,.
\]
The relative size of the product $(1-\tau) x$ on the right hand side and $1$ 
decides if the expansion of the left hand side is done around $\beta=0$ 
or $\beta \to \infty$. We recall that a similar result holds for the equation 
\eqref{eq-z-xi}
determining $\xi$ for $K<S_0$
\[
(1-\tau)\log\left(\frac{\sin 2\xi}{2\xi}\right)+2\tau\xi\tan\left(\xi\right)
=(1-\tau)x\,.
\]

We consider the asymptotic expansion of the rate function for
forward starting Asian options in the Black-Scholes model in each of these
regions.
We first consider the asymptotics of the rate function in the
$\tau-$AATM region $(1-\tau) |x| \ll 1$.

\begin{proposition}\label{prop:3.10}
In the $\tau-$AATM region $(1 - \tau)|x| \ll 1$ we have the following 
asymptotic expansion in
powers of $\frac{1-\tau}{\tau} x$ in the Black-Scholes model.
\begin{align*}
& \I_{\rm fwd}^{(BS)}(x, \tau, \sigma)= 
\frac{1}{(1-\tau)^2\sigma^2}\bigg\{
\frac{3\tau^2}{2(1+2\tau)}\left(\frac{x(1-\tau)}{\tau} \right)^2
-\frac{3\tau^3(1-\tau)}{10(1+2\tau)^3}\left(\frac{x(1-\tau)}{\tau} \right)^3
\\
&\qquad\qquad
+\frac{\tau^4(109-349\tau)}{1400(1+2\tau)^5}\left(\frac{x(1-\tau)}{\tau} \right)^4
+O\left( \left({x(1-\tau)}\right)^5\right)
\bigg\}.
\end{align*}
\end{proposition}

\begin{proof}
(1) Consider first the case $K>S_0$. 
We have
\begin{equation*}
(1-\tau)\log\left(\frac{\sinh \beta}{\beta}\right)+\tau\beta\tanh\left(\frac{\beta}{2}\right)\ll1.
\end{equation*}
This implies that $\beta\ll1$. 

From \eqref{eq-z-beta} we have
\begin{align*}
\frac{1-\tau}{\tau}x&=\beta\tanh \frac{\beta}{2}+\frac{1-\tau}{\tau}\log\left(\frac{\sinh \beta}{\beta} \right)\\
&=\frac{\beta^2}{2}-\frac{\beta^4}{24}+O(\beta^6)+\frac{1-\tau}{\tau}\left( \frac{\beta^2}{6}-\frac{\beta^4}{180}+O(\beta^6)\right)
\end{align*}
Hence
\begin{eqnarray*}
\frac12\beta^2 = \frac{3\tau}{1+2\tau} \left( \frac{1-\tau}{\tau} x\right) +
\frac{3\tau^2(2+13\tau)}{10(1+2\tau)^3} \left( \frac{1-\tau}{\tau} x\right)^2 +
O\left(\left(\frac{1-\tau}{\tau}x \right)^3\right)
\end{eqnarray*}

and 
\begin{align*}
\I_{\rm fwd}^{(BS)}(x, \tau, \sigma)
&=\frac{1}{2\sigma^2(1-\tau)^2}\left( \frac{(1+2\tau)\beta^4}{12}-\frac{1+4\tau}{120}\beta^6+\frac{17(1+6\tau)}{20160}\beta^8+O(\beta^{10})\right)\\
&=\frac{3\tau^2}{2\sigma^2(1-\tau)^2(1+2\tau)} \left( \frac{1-\tau}{\tau} x\right)^2-\frac{3\tau^3}{10\sigma^2(1-\tau)(1+2\tau)^3} \left( \frac{1-\tau}{\tau} x\right)^3
\\
&\qquad\qquad 
+\frac{\tau^4(109-349\tau)}{1400\sigma^2(1-\tau)^2(1+2\tau)^5}\left(\frac{1-\tau}{\tau}x \right)^4+O\left(\left(\frac{1-\tau}{\tau}x \right)^5\right).
\end{align*}

(2) Next consider $K<S_0$. When $-(1-\tau)x\ll1$, from \eqref{eq-z-xi} we know 
that $\xi\ll 1$, and by expanding it we obtain
\begin{equation*}
-\frac{x(1-\tau)}{\tau}=\frac{2(1+2\tau)}{3\tau}\xi^2+\frac{2(2+13\tau)}{45\tau}\xi^4+O(\xi^6),
\end{equation*}
hence we have
\begin{eqnarray*}
\xi^2 = \frac{3\tau}{2(1+2\tau)} \left( - \frac{1-\tau}{\tau} x\right) -
\frac{3\tau^2(2+13\tau)}{20(1+2\tau)^3} 
\left( - \frac{1-\tau}{\tau} x\right)^2 + O((-(1-\tau)x)^3)
\end{eqnarray*}

Therefore we obtain the asymptotic expansion of the rate function
\begin{align*}
\I_{\rm fwd}^{(BS)}(x, \tau, \sigma)
&=\frac{2}{\sigma^2(1-\tau)^2}\left(
\frac{1+2\tau}3\xi^4+\frac{2(4\tau+1)}{15}\xi^6+\frac{17(1+6\tau)}{315}\xi^8+O(\xi^{10})
\right)
\\
&
=\frac{1}{\sigma^2(1-\tau)^2}\bigg\{
\frac{3\tau^2}{2(1+2\tau)}\left(\frac{-x(1-\tau)}{\tau} \right)^2+
\frac{3\tau^3(1-\tau)}{10(1+2\tau)^3}\left(\frac{-x(1-\tau)}{\tau} \right)^3
\\
&\qquad\qquad
+\frac{\tau^4(109-349\tau)}{1400(1+2\tau)^5}\left(\frac{-x(1-\tau)}{\tau} \right)^4+O\left( \left({-x(1-\tau)}\right)^5\right)
\bigg\}.
\end{align*}
Hence the conclusion follows.
\end{proof}

\begin{remark}
Recall that when $\log(K/S_0)=o(1)$, the rate function of a standard Asian 
option with averaging starting time zero ${\cal J}_{\rm BS}(K/S_0)$ is  approximately
$\frac32 \log^2\left(\frac{K}{S_0}\right)$ (see \cite{PZAsian}, page 7). Plug it into \eqref{eq-c} and  find the minimum of the quadratic equation of $c$
\begin{equation*}
\frac{1}{2}\tau c^2+\frac{3}{2(1-\tau)}\left( \log\left(\frac{K}{S_0} \right)-c\tau \right)^2.
\end{equation*}
We can easily compute the minimizer $c_*=\frac{3}{1+2\tau}\log\left(\frac{K}{S_0} \right)$ and the minimum $ \frac{3}{2(1+2\tau)} \log^2 \left(\frac{K}{S_0}\right)$. This provides us a simple way to estimate the rate function of an Asian forward option, which indeed coincides with the first term of \eqref{eq-BS-ATM-I}. However, this approach is not accurate enough to provide further terms in the asymptotic expansion.
\end{remark}

Since AATM region is included in the $\tau$--AATM region, it follows
immediately from Proposition~\ref{prop:3.10} that we have:

\begin{corollary}
In the AATM region  $x\to0$ we have the following expansion for the rate function
of a forward start Asian option under the Black-Scholes model with parameters 
$(x=\log(K/S_0), \tau, \sigma)$
\begin{align}\label{eq-BS-ATM-I}
&\I_{\rm fwd}^{(BS)}(x, \tau, \sigma)=\frac{3}{2\sigma^2}\bigg\{
\frac{x^2}{1+2\tau}-\frac{(1-\tau)^2}{5(1+2\tau)^3}x^3 \\
& \qquad \qquad \qquad\qquad+ \frac{(1-\tau)^3(109 - 349\tau)}{2100(1+2\tau)^5} x^4+O(x^5)
\bigg\}.\nonumber
\end{align}
\end{corollary} 

\begin{remark}
In the limit $\tau=0$ this reduces to 
\begin{equation}\label{JTaylor}
\mathcal{I}_{\rm fwd}^{(BS)}(x,0, \sigma) = \frac{1}{\sigma^2}\left\{
\frac32 x^2 - \frac{3}{10} x^3 +
\frac{109}{1400} x^4 + O(x^5) \right\}\,,
\end{equation}
which reproduces Eq.(35) in \cite{PZAsian}.
 In the opposite limit $\tau \to 1$ at $x\sim O(1)$ we get 
from proposition ~\ref{prop:3.10} the simple limit
\begin{equation}
\mathcal{I}_{\rm fwd}^{(BS)}(x,1,\sigma) = \frac{1}{2\sigma^2} x^2 \,,
\end{equation}
which is the rate function for the large deviations of a 
sum of iid Gaussian random variables, which gives also the rate function for
an out-of-the-money European option in the Black-Scholes model.
\end{remark}

Next let us consider the asymptotic behaviors of the rate functions in the 
$\tau-$deep out-of-the-money regime. 

\begin{proposition}
Consider an out-of-the-money forward start Asian call option under the
Black-Scholes model with parameter $(x, \tau, \sigma)$. 
When the option is in the
deep out-of-the-money region, namely $(1-\tau) x\gg1$, its rate function has 
the following 
asymptotic expansion
\begin{align}\label{eq-asymp-x}
\I_{\rm fwd}^{(BS)}(x, \tau, \sigma)
&=\frac{1}{2\sigma^2}\bigg\{ 
x^2+2x\log(2(1-\tau)x)-2x+\left(\log(2(1-\tau)x)\right)^2
\\
&\qquad\qquad+O\bigg(\frac{\log^2(2(1-\tau)x)}{x}\bigg) 
\bigg\}.
\nonumber
\end{align}
\end{proposition} 

\begin{proof}
Clearly when $x=\log(K/S_0)\gg1$, from \eqref{eq-z-beta} we have
\begin{equation*}
\frac{\sinh\beta}{\beta}=e^{x-\frac{\tau}{1-\tau}\beta\tanh\frac{\beta}{2}}.
\end{equation*}

We have $x=\log(K/S_0)\gg1$ and $(1-\tau)x\gg 1$. 
From the above equation we know that 
\begin{equation*}
(1-\tau)\log\left(\frac{\sinh \beta}{\beta}\right)+\tau\beta\tanh\left(\frac{\beta}{2}\right)\gg1.
\end{equation*}
This implies that $\beta\gg1$. Now from \eqref{eq-z-beta} we obtain
\begin{equation*}
x=\beta-\log(2\beta)+\frac{\tau}{1-\tau}\beta+O\left(e^{-2\beta}+\frac{\tau}{1-\tau}e^{-\beta}\right).
\end{equation*}
By inverting the series we have
\begin{equation*}
\beta=(1-\tau)x+(1-\tau)\log(2(1-\tau)x)+(1-\tau)\frac{\log(2(1-\tau)x)}{x}+O\bigg((1-\tau)\frac{\log^2(2(1-\tau)x)}{x^2}\bigg).
\end{equation*}
By plugging in \eqref{eq-BS-rate} we have
\begin{align*}
&\I_{\rm fwd}^{(BS)}(x, \tau, \sigma)
\\
&=\frac{1}{2\sigma^2(1-\tau)^2}\left(\beta^2-2(1-\tau)\beta+O(e^{-\beta})\right)
\\
&=\frac{1}{2\sigma^2}\bigg\{ 
x^2+2x\log(2(1-\tau)x)-2x+\log^2(2(1-\tau)x)+O\bigg(\frac{\log^2(2(1-\tau)x)}{x}\bigg) 
\bigg\}.
\end{align*}
\end{proof}

\begin{remark}
Consider the fixed $x$ and $\tau\to0$ limit. 
By taking the limit $\tau\to0$ in \eqref{eq-asymp-x} 
we obtain
\begin{align*}
&\I_{\rm fwd}^{(BS)}(x, \tau, \sigma)
 = \frac{1}{\sigma^2}\bigg\{ 
\frac12 x^2+x\log(2x)-x+
\frac12 \log^2(2x) +O\bigg(\frac{\log^2(2x)}{x}\bigg) 
\bigg\}\,.
\end{align*}
This can be compared with the large strike asymptotic result for the rate 
function of a standard Asian option with averaging starting from time $0$ 
in Eq.~(36) of \cite{PZAsian}\footnote{
Note that in Eq.~(36) of \cite{PZAsian} the coefficient
of $\log^2(2x)$ should be $\frac12$ instead of $\frac32$, and there is
no $\log(2x)$ term to the order considered. With these corrections, the
$\tau =0$ result is reproduced indeed.}.
\end{remark}

Next we consider the case of a deep out-of-the-money put Asian option,  
which has small strike, namely $x=\log(K/S_0)\ll -1$. This case is more
complex, and we distinguish two regions in $(\tau,x)$ with distinct
asymptotics.

\begin{proposition}\label{prop:LeftWing}
Let $\I_{\rm fwd}^{(BS)}(x, \tau, \sigma)$ be the rate function of an 
out-of-the-money 
forward start put Asian option under the Black-Scholes model. 
Assume the put option is in the deep out-of-the-money region, 
namely $(1-\tau) x\ll-1$. 
\begin{itemize}
\item[(1)] If $-(1-\tau)x\gg 1$ and $\frac{2\tau}{1-\tau}\ll e^x(-x)$, namely $\frac{e^x(-x)}{\tau}\to+\infty$, we have $\tau\to0$ and
\begin{equation}\label{eq-tau-0-asymp}
\I_{\rm fwd}^{(BS)}(x, \tau, \sigma)
=\frac{2}{\sigma^2(1-\tau)^2}\bigg\{
e^{-2x}\tau+(1-3\tau)e^{-x}-\frac{\pi^2}{4}-1 +o(e^{-2x}\tau+e^{-x})\bigg\}.
\end{equation}
\item[(2)] If $-(1-\tau)x\gg 1$ and $\frac{2\tau}{1-\tau}\gg e^x(-x)$, 
namely $\frac{e^x(-x)}{\tau}\to0$, we have
\begin{align*}
\I_{\rm fwd}^{(BS)}(x, \tau, \sigma)
&
= \frac{2}{\sigma^2}\bigg\{
\frac{x^2}{4\tau}-\frac{-x}{2\tau}\log\left( \frac{-x(1-\tau)}{2\tau}\right)
+\frac{-x}{2\tau} \\
&\qquad\qquad+\frac{1}{4\tau}\log^2\left( \frac{-x(1-\tau)}{2\tau}\right)
-\frac{\pi^2}{12}\frac{3-\tau}{(1-\tau)^2}+o(1)\bigg\}.\nonumber
\end{align*}
\end{itemize}
\end{proposition} 

\begin{remark}
Case (1) corresponds to the left wing asymptotics $(-x) \to \infty$ in 
a region of small $\tau$, below the curves shown in Fig.~\ref{Fig:regions}.
Case (2) corresponds to $\tau \sim O(1)$ and large log-strike $-x \gg 1$.
\end{remark}

\begin{figure}[ht]
\begin{center}
\includegraphics[height=60mm]{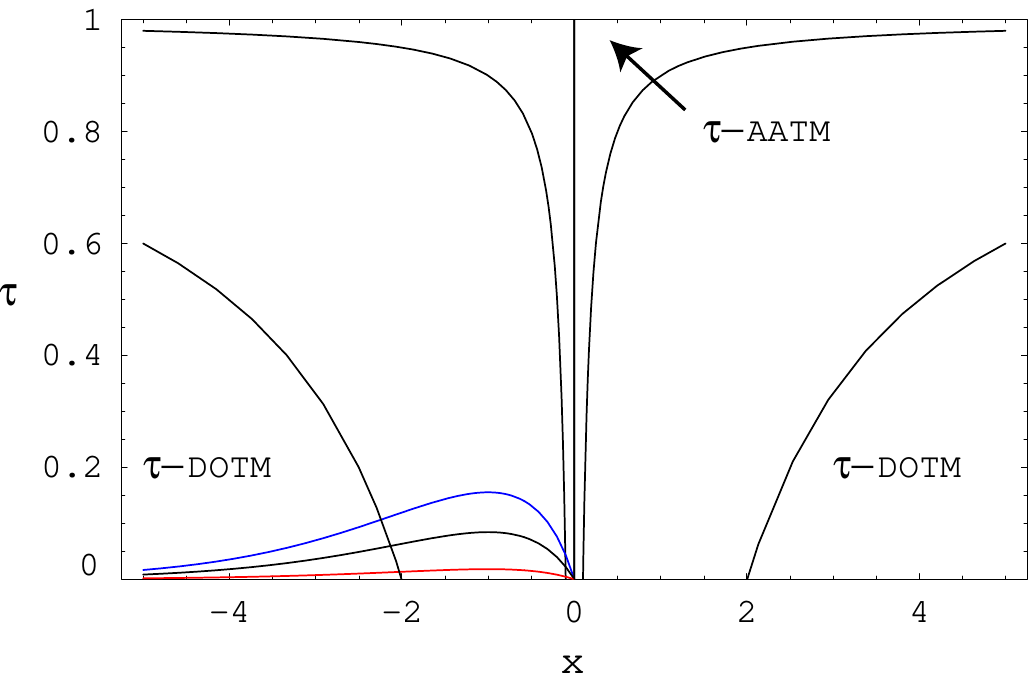}
\end{center}
\caption{
Regions in the $(x,\tau)$ plane with distinct asymptotic expansion.
The interior of the central funnel-shaped region corresponds to the 
$\tau$-AATM region (represented as $(1-\tau)|x| \leq 0.1$).
The region $\tau-$DOTM corresponds to  $(1-\tau)|x|\gg 1$
(represented as $(1-\tau)|x| \geq 2$).
The colored curves at $x<0$ are curves of constant $\kappa$ defined as 
$\frac{2\tau}{1-\tau} = \kappa e^x (-x)$:
$\kappa=1$ (blue), 0.5 (black), 0.1 (red). 
Region (1) of Proposition~\ref{prop:LeftWing} corresponds to values of 
$(x,\tau)$ below 
these curves which are also in the $\tau-$DOTM region, and 
region (2) corresponds to values of 
$(x,\tau)$ above these curves which are also in the $\tau-$DOTM region.}
\label{Fig:regions}
\end{figure}

\begin{proof}
Clearly when $x<0$, from  \eqref{eq-z-xi}  we have $\xi>0$ and 
\begin{equation*}
-x=\frac{2\tau}{1-\tau}\xi\tan \xi+\log\left(\frac{2\xi}{\sin 2\xi}\right).
\end{equation*}
(1) Since $-(1-\tau)x\gg 1$ we know that $\xi\to \frac{\pi}{2}$. Let $\zeta=\frac{\pi}{2}-\xi$ where $\zeta\to 0$. Expand the above equation at $\zeta=0$ we obtain
\begin{equation}\label{eq-x-zeta}
-x=\frac{2\tau}{1-\tau}\frac{\pi}{2\zeta}+\log\left(\frac{\pi}{2\zeta}\right)-\frac{2\tau}{1-\tau}+O(\zeta).
\end{equation}
Now since $\frac{2\tau}{1-\tau}\ll e^x(-x)$, we claim that 
\begin{equation}\label{eq-claim-1}
\frac{2\tau}{1-\tau}\frac{\pi}{2\zeta}\ll \log\left(\frac{\pi}{2\zeta}\right).
\end{equation} 
Otherwise if  $\frac{2\tau}{1-\tau}\frac{\pi}{2\zeta}\gtrsim \log\left(\frac{\pi}{2\zeta}\right)$\footnote{For $a, b>0$, we say $a\gtrsim b$ if there exists a constant $C>0$ such that $a\ge Cb$.},  \eqref{eq-x-zeta} implies that $-x\asymp\frac{2\tau}{1-\tau}\frac{\pi}{2\zeta}$\footnote{We denote by $a\asymp b$ if there exists constant $C_1, C_2>0$ such that $C_1a\le b\le C_2a$.}. Plug this back into the previous inequality we obtain $
 -x\gtrsim\log \left(\frac{1-\tau}{2\tau}(-x) \right)$,
 which conflicts with the assumption $\frac{2\tau}{1-\tau}\ll e^x(-x)$. 
Hence claim \eqref{eq-claim-1} holds. Plug into \eqref{eq-x-zeta}, we have
\begin{equation*}
-x=\log\left(\frac{\pi}{2\zeta}\right)+o(\log({1}/{\zeta})).
\end{equation*}
Hence the asymptotic expansion of the rate function is 
\begin{align*}
\I_{\rm fwd}^{(BS)}(x, \tau, \sigma)
&=\frac{2}{\sigma^2(1-\tau)^2}\left(
\tau\left(\frac{\pi}{2\zeta} \right)^2+(1-3\tau)\frac{\pi}{2\zeta}-\frac{\pi^2}{4}-1 +O(\tau+\zeta)
\right)
\\
&= \frac{2}{\sigma^2(1-\tau)^2}\bigg\{
e^{-2x}\tau+(1-3\tau)e^{-x}-\frac{\pi^2}{4}-1 +o(e^{-2x}\tau+e^{-x})\bigg\}.
\end{align*}
(2) Again we have $\xi\to \frac{\pi}{2}$ and $\zeta=\frac{\pi}{2}-\xi$, $\zeta\to 0$. When $\frac{2\tau}{1-\tau}\gg e^x(-x)$, \eqref{eq-x-zeta} implies that
\begin{equation}\label{eq-claim-2}
\frac{2\tau}{1-\tau}\frac{\pi}{2\zeta}\gg \log\left(\frac{\pi}{2\zeta}\right).
\end{equation} 
The proof of the above claim is similar as in (1) by reversing the signs. 
Hence
\begin{equation*}
-x=\frac{2\tau}{1-\tau}\frac{\pi}{2\zeta}+\log\left(\frac{\pi}{2\zeta}\right)+O(\tau+\zeta).
\end{equation*}

By inverting the series we have
\begin{equation*}
\frac{\pi}{2\zeta}=\frac{-x(1-\tau)}{2\tau}-
\frac{1-\tau}{2\tau}\log\left(\frac{-x(1-\tau)}{2\tau}\right)
+1+\frac{1-\tau}{2\tau}\cdot\frac{\log\left(\frac{-x(1-\tau)}{2\tau}\right)}{-x}+o\left(
\frac{\log\left(\frac{-x(1-\tau)}{2\tau}\right)}{-x}\right).
\end{equation*}
Hence the asymptotic expansion of the rate function is 
\begin{align*}
&\I_{\rm fwd}^{(BS)}(x, \tau, \sigma)
\\
&=\frac{2}{\sigma^2(1-\tau)^2}\left(
\tau\left(\frac{\pi}{2\zeta} \right)^2+(1-3\tau)\frac{\pi}{2\zeta}+\left(\frac{1}{12}(24+\pi^2)\tau-\frac{\pi^2}{4}-1 \right) +O(\zeta)
\right)
\\
& = 
\frac{2}{\sigma^2}\bigg\{
\frac{x^2}{4\tau}-\frac{-x}{2\tau}\log\left( \frac{-x(1-\tau)}{2\tau}\right)+\frac{-x}{2\tau}+\frac{1}{4\tau}\log^2\left( \frac{-x(1-\tau)}{2\tau}\right)
-\frac{\pi^2}{12}\frac{3-\tau}{(1-\tau)^2}+o(1)\bigg\}.
\end{align*}

\end{proof}

\begin{remark}
For a fixed $x$ if $\tau\to0$, it falls into the first case that  $-(1-\tau)x\gg 1$ and $\frac{2\tau}{1-\tau}\ll e^x(-x)$. By taking \eqref{eq-tau-0-asymp} to the limit $\tau\to0$ we obtain
\begin{align*}
&\I_{\rm{BS}}(x, \tau, \sigma)
 = \frac{2}{\sigma^2}\left(
e^{-x}-\frac{\pi^2}{4}-1 +o(1)
\right),
\end{align*}
which agrees with the asymptotic result for a standard Asian option 
with averaging starting from time $0$ (see \cite{PZAsian}, page 8).
\end{remark}

\begin{figure}[ht]
\begin{center}
\includegraphics[height=60mm]{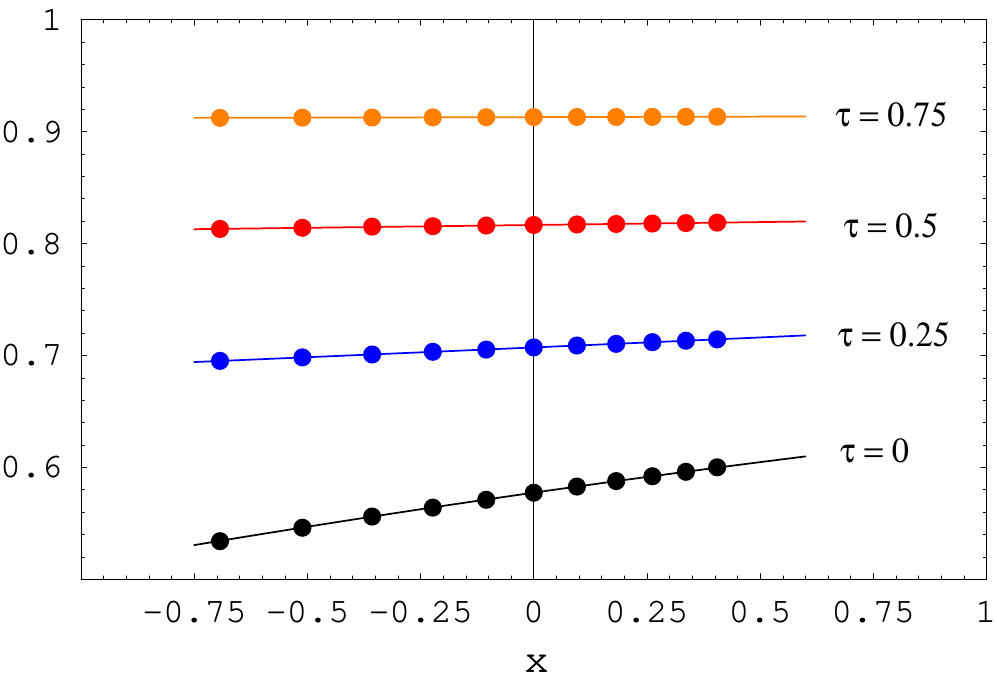}
\end{center}
\caption{The equivalent log-normal volatility $\Sigma_{\rm LN}(x,\tau)/\sigma$ 
of the forward start Asian option in the Black-Scholes model, 
vs $x=\log(K/S_0)$, for $\tau=0, 0.25,
0.5,0.75$ (from bottom to top). The dots show numerical values obtained
from Eq.~(\ref{SigLN}) and the solid curves show the numerical approximation 
obtained keeping three terms in the series expansion approximation 
(\ref{SigLNTaylor}).}
\label{Fig:4}
\end{figure}

\subsection{Equivalent Black--Scholes volatility}
\label{sec:3.4}
It is convenient to introduce the so-called equivalent log-normal 
(Black--Scholes) volatility $\Sigma_{\rm LN}$ of a forward start Asian option. 
This is defined as the constant volatility for which the Black--Scholes price 
of an European (vanilla) option with maturity $T$ and underlying value 
\begin{equation}\label{Afwddef}
A(\tau T,T)=\frac{1}{(1-\tau)T}\int_{\tau T}^{T}\mathbb{E}(S_{t})dt =
\frac{S_0}{(r-q)(1-\tau)T} (e^{(r-q)T} - e^{(r-q)\tau T})
\end{equation}
reproduces the price of the forward start Asian option with parameters
$(S_0, K,T,\tau)$, which we denote by $C(T)$ ($P(T)$ resp.). Analogously we
define the normal (Bachelier) equivalent volatility as that volatility
$\Sigma_N$ for which the Bachelier price of an European option with maturity
$T$  and strike $K$ reproduces the Asian option price.

Take an Asian call option as example. We have
\begin{equation*}
C_\BS(S_0, K, \Sigma_{\rm LN}, T)=C(T),
\end{equation*}
where
\begin{equation*}
C_\BS(S_0, K, \Sigma_{\rm LN}, T)=e^{-rT}\left(A(T) \Phi(d_1)-K\Phi(d_2)\right)\,,
\end{equation*}
with $d_{1,2}=\frac{1}{\Sigma_{\rm LN}\sqrt{T}}
\left(\log (A(\tau T,T)/K)\pm\frac12\Sigma^2_{\rm LN} T \right)$,
where $A(\tau T,T)$ is given in (\ref{Afwddef}) and
\begin{equation*}
C(T)=e^{-rT}\mathbb{E}\left[\left(\frac{1}{(1-\tau)T}\int_{\tau T}^{T}S_{t}dt-K\right)^{+}\right]\,,
\end{equation*}
with $S_t$ following the Black--Scholes model with volatility $\sigma$.

We have the analog of Proposition 18 in \cite{PZAsian}.

\begin{proposition}\label{PropEquivVol}
Assume $r=q=0$ and the stated assumptions on the local volatility function. 

(1) The short-maturity $T\to 0$ limit of the log-normal equivalent 
volatility
of a forward start out-of-the-money Asian option is
\begin{equation}\label{SigLN}
\lim_{T\to 0} \Sigma_{\rm LN}^2(K,S_0,T) = 
\frac{\log^2(K/S_0)}{2\mathcal{I}_{\rm fwd}(K,S_0,\tau)}\,.
\end{equation}
The corresponding result for the normal (Bachelier) equivalent volatility 
is\footnote{We correct here a typo in Eq.~(52) of \cite{PZAsian}, which 
should have $(K-S_0)^2$ in the numerator.}
\begin{equation}\label{SigN}
\lim_{T\to 0} \Sigma_{\rm N}^2(K,S_0,T) = 
\frac{(K-S_0)^2}{2\mathcal{I}_{\rm fwd}(K,S_0,\tau)}\,.
\end{equation}

(2) The short-maturity $T\to 0$ limit of the log-normal equivalent 
volatility of a forward start ATM Asian option is
\begin{equation}\label{SigLNATM}
\lim_{T\to 0} \Sigma_{\rm LN}(S_0,S_0,T) = 
\sigma(S_0) \sqrt{\frac{1+2\tau}{3}}\,.
\end{equation}
\end{proposition}

\begin{remark}
Substituting here the expansion of the forward start rate function in the
Black-Scholes model
and expanding in powers of $x=\log(K/S_0)$ we get
\begin{equation}\label{SigLNTaylor}
\lim_{T\to 0} \Sigma_{\rm LN}(K,S_0,T) =\sigma \sqrt{\frac{1+2\tau}{3}} 
\left( 1 + \frac{(1-\tau)^2}{10(1+2\tau)^2} x 
- \frac{(1-\tau)^3 (23-143\tau)}{2100(1+2\tau)^4} x^2 + O(x^3) \right)\,.
\end{equation}
In the limit $\tau \to 0$ this agrees with the result in
Eq.~(55) in \cite{PZAsian} for the log-normal implied volatility of an Asian
option with averaging starting at time 0 in the Black-Scholes model. 
The leading term gives the ATM
volatility, the $O(x)$ term gives the ATM skew, and the
$O(x^2)$ term gives the ATM curvature of the Asian implied volatility smile. 
\end{remark}

\begin{proof}[Proof of Proposition~\ref{PropEquivVol}]
The proof is analogous to the proof of Proposition 18 in \cite{PZAsian}
and is omitted here.
\end{proof}

\begin{remark}
The result in (\ref{SigLNATM}) corresponds to the log-normal implied
volatility of an at-the-money forward start Asian option paying 
$(A_{[T_1,T_2]} - K)^+$ at time $T_2$.  
From this result we observe that the price of this option is the same as 
that of an European option\footnote{This result is informally known to quants
as a general rule of thumb for pricing forward start Asian options \cite{thumb}.} 
with variance
$\sigma^2 (T_1 + \frac13 (T_2- T_1))$. The variance can be also written as
\begin{equation}
\sigma^2 \left(\tau T + \frac13 (1-\tau) T \right) = \frac13 (1 + 2\tau) \sigma^2 T
= \Sigma_{\rm LN}^2(S_0,S_0) T\,,
\end{equation}
we get that this is also equal to the price of an European option with
volatility $\Sigma_{\rm LN}(S_0,S_0) := \lim_{T\to 0} \Sigma_{\rm LN}(S_0,S_0,T)$ 
and maturity $T$. 
\end{remark}


\section{Floating strike Asian options}\label{sec-float}

An important class of Asian options 
that are encountered in practice are the floating Asian 
options. The strike of these options is defined with respect to the average
of the asset price over a forward period $[T_1,T_2]$ with $T_1 < T_2$.

Such an option pays at time $T_2$ the payoff of the form
\begin{eqnarray}
\mbox{Call} &:& \left(
\kappa S_{T_2} - \frac{1}{T_2-T_1} \int_{T_1}^{T_2} S_t dt \right)^+\\
\mbox{Put} &:& \left(
\frac{1}{T_2-T_1} \int_{T_1}^{T_2} S_t dt - \kappa S_{T_2}\right)^+\,.
\end{eqnarray}
Here $\kappa>0$ is the so-called floating strike.


Denoting $t$ the valuation time, we distinguish two cases for $t$:

(i) $t < T_1$. The entire averaging period is in the future. The option
is forward start and this case
is studied in Sec.~\ref{sec:4.1}.

(ii) $T_1 < t < T_2$. Namely the valuation time $t$ enters into the averaging 
period, and the
integral over the asset price includes a deterministic contribution and a
stochastic contribution 
\begin{equation}
\int_{T_1}^{T_2} S_u du = \int_{T_1}^t S_u du + \int_t^{T_2} S_u du\,.
\end{equation}
Part of the averaging period is deterministic,
and can be absorbed into a constant strike in the payoff. This corresponds
to so-called ``seasoned Asian option'', or generalized Asian options.
This case will be treated below in Sec.~\ref{sec:4.2}.

\subsection{Forward start floating strike Asian options}
\label{sec:4.1}

We consider in this section the forward start floating
strike Asian options with $t < T_1 < T_2$, and take $t=0$ for simplicity.
Let $T_2 = T >0$, $T_1 = \tau T$ with $\tau\in(0,1)$. 
The prices of these options are given by expectations in the risk-neutral
measure
\begin{align*}
&C_f(\kappa, T,\tau )=e^{-rT}\E\bigg[\left(\kappa S_T-\frac{1}{(1-\tau)T}\int_{\tau T}^TS_tdt \right)^+ \bigg],\\
&P_f(\kappa, T,\tau )=e^{-rT}\E\bigg[\left(\frac{1}{(1-\tau)T}\int_{\tau T}^TS_tdt-\kappa S_T \right)^+ \bigg].
\end{align*}
We will assume that the underlying asset price $S_t$ follows the local volatility 
model \eqref{eq-S-t} and \eqref{eq-S-t-condition}.

The floating strike call Asian option is in-the-money (floating strike
put Asian option is
out-of-the-money) when $\kappa S_0 e^{(r-q)T} > A(\tau T,T)$
with $A(\tau T,T) = \frac{1}{(r-q)(1-\tau)T}(e^{(r-q)T} - e^{(r-q)\tau T})$, and 
they are both at-the-money when $\kappa S_0 e^{(r-q)T} = A(\tau T,T)$.
In the $T\to 0$ limit we have $\lim_{T\to 0} A(\tau T,T) = S_0$ and we get that
for $\kappa<1$, the call option is out-of-the-money and the put option is 
in-the-money; when  $\kappa=1$, both the call and put options are at-the-money; 
when $\kappa>1$, the call option is in-the-money and the put option is 
out-of-the-money.

Similar to the case of the fixed strike Asian options, as the maturity $T\to0$, 
the prices of out-of-the-money floating-strike forward start Asian options 
decrease to $0$ exponentially fast, and the rate can be captured using large 
deviation theory. The prices of at-the-money options, on the other hand, 
decrease at the speed of $\sqrt{T}$, and the rates can be obtained by linear 
approximation with a Gaussian random variable. Also as before, the estimates 
on the prices of in-the-money floating-strike forward start Asian options 
follow easily from put-call parity. In the following theorem, we study these 
asymptotic estimates in detail.

\begin{theorem}\label{thm-f-otm}
Assume the asset price $S_t$ follows the local volatility model as in 
\eqref{eq-S-t} and \eqref{eq-S-t-condition}. 
Then we have the following short maturity asymptotic estimates for the 
prices of floating-strike forward start Asian options as $T\to0$. 

(1) When $\kappa<1$, we have
\begin{equation}\label{eq-f-out-call}
\lim_{T\to0}T\log C_f(\kappa, T,\tau)=-{\cal I}_{f}(S_0,\kappa, \tau ),
\end{equation}
and
\begin{equation}\label{eq-f-in-put}
P_f(\kappa, T,\tau)=S_0(1-\kappa)-S_0T\left(\frac12(r+q)-\kappa q-\frac12(r-q)\tau \right)+O(T^2)\,,
\end{equation}
where
\begin{equation*}
{\cal I}_{f}(S_0,\kappa, \tau)=
\inf_{
\substack{\frac{1}{1-\tau}\int_{\tau}^{1}e^{g(t)}dt=\kappa e^{g(1)}
\\
g(0)=\log S_{0}, g\in\mathcal{AC}[0,1]}}\frac{1}{2}\int_{0}^{1}\left(\frac{g'(t)}{\sigma(e^{g(t)})}\right)^{2}dt.
\end{equation*}

(2) When $\kappa>1$, we have
\begin{equation}\label{eq-f-in-call}
C_f(\kappa, T,\tau)=S_0(\kappa-1)+S_0T\left(\frac12(r+q)-\kappa q-\frac12(r-q)\tau \right)+O(T^2)\,,
\end{equation}
and
\begin{equation}\label{eq-f-out-put}
\lim_{T\to0}T\log P_f(\kappa, T,\tau)=-{\cal I}_{f}(S_0, \kappa, \tau ).
\end{equation}

(3)  When $\kappa=1$, we have
\begin{equation}\label{eq-atm}
\lim_{T\rightarrow 0}\frac{1}{\sqrt{T}}C_f(\kappa, T,\tau)
=\lim_{T\rightarrow 0}\frac{1}{\sqrt{T}}P_f(\kappa, T,\tau)
=\frac{\sigma(S_{0})S_{0}}{\sqrt{6\pi}} \sqrt{1-\tau}\,.
\end{equation}
\end{theorem}
\begin{proof}
(1) Following the same argument as in Lemma \ref{lemma-TCT}, we immediately obtain that
\begin{equation*}
\lim_{T\rightarrow 0}T\log C_f(\kappa, T,\tau)
=\lim_{T\rightarrow 0}T\log\mathbb{P}\left(\kappa S_T\ge\frac{1}{1-\tau}\int_{\tau}^{1}S_{tT}dt\right).
\end{equation*}
Then by applying the sample path large deviation of $\p(S_{t\cdot}\in\cdot, t\in [0,1])$ on $L_\infty[0,1]$ and the contraction principle, we obtain \eqref{eq-f-out-call}.
To see \eqref{eq-f-in-put}, we just need to use put-call parity
\begin{align*}
C_f(\kappa, T,\tau)&-P_f(\kappa, T,\tau)=e^{-rT}\left(\kappa \E(S_T)-\frac{1}{(1-\tau)T}\int_{\tau T}^T\E(S_t)dt \right)\\
&=e^{-rT}S_0\left(\kappa e^{(r-q)T}-e^{(r-q)\tau T}\left(1+\frac12(r-q)(1-\tau)T+O(T^2)\right) \right)\\
&=S_0(\kappa-1)+S_0T\left(\frac12(r+q)-\kappa q-\frac12(r-q)\tau \right)+O(T^2).
\end{align*}
From \eqref{eq-f-out-call} we know that 
$C_f(\kappa, T,\tau)=O(e^{-{\cal{I}}_f/T})=o(T^2)$ when $\kappa<1$, then 
\eqref{eq-f-out-put} follows immediately.\\
(2) We can easily obtain \eqref{eq-f-in-call} and \eqref{eq-f-out-put} using the same arguments as in (1).\\
(3) Following similar arguments as in Theorem \ref{thm-atm}, we have
\begin{equation*}
\left|C_f(\kappa, T,\tau)-\mathbb{E}\left[\left(\kappa \hat{S}_T-\frac{1}{1-\tau}\int_{\tau }^{1}\hat{S}_{tT}dt\right)^{+}\right]\right|
=O(T),
\end{equation*}
where $\hat{S}_{t}$ is a Gaussian process given by $\hat{S}_t=S_{0}+\sigma(S_{0})S_{0}W_{t}$.  When $\kappa=1$, since $\hat{S}_T-\frac{1}{1-\tau}\int_{\tau }^{1}\hat{S}_{tT}dt$ is a normal random variable with mean zero and  variance
\begin{align*}
&\sigma(S_{0})^{2}S_{0}^{2}\,
\mathbb{E}\left[\left(W_T-\frac{1}{1-\tau}\int_{\tau}^{1}W_{tT}dt\right)^{2}\right]
=\frac{1}{3}(1-\tau) \sigma(S_{0})^{2}S_{0}^{2}T.
\nonumber
\end{align*}

Hence we have
\begin{equation*}
\mathbb{E}\left[\left(\hat{S}_T-\frac{1}{1-\tau}\int_{\tau }^{1}\hat{S}_{tT}dt\right)^{+}\right]
=\frac{1}{\sqrt{3}} \sqrt{1-\tau}
\sigma(S_{0})S_{0}\sqrt{T}\,\mathbb{E}[Z\mathbbm{1}_{Z>0}],
\end{equation*}
where $Z$ is a standard normal random variable. We then obtain \eqref{eq-atm}.
\end{proof}

In the rest of this section we further investigate the variational problems of the rate functions.
\begin{proposition}\label{prop-f-rate}
Assume the asset price $S_t$ follows the local volatility model as in \eqref{eq-S-t} and \eqref{eq-S-t-condition}.  Consider an out-of-the-money floating-strike forward start Asian option and let ${\cal I}_{f}(S_0,\kappa, \tau)$ be the rate function as in Theorem \ref{thm-f-otm}. We have
\begin{equation}\label{eq-rate-float-s}
{\cal I}_{f}(S_0,\kappa, \tau)= \inf_{c\in \mathbb{R}}
\left\{ \frac12 c^2 \tau + \frac{1}{1-\tau}
{\cal I}\left(S_0 e^{F^{-1}(c\tau)}, \kappa \right) \right\},
\end{equation}
where $F(\cdot) = \int_0^{\cdot} \frac{dz}{\sigma(S_0 e^z)}$ and
\begin{equation}\label{eq-cal-I}
{\cal I}\left(x, \kappa \right)  = \inf_{
\substack{\int_{0}^{1}e^{\varphi(u)}du=\kappa e^{\varphi(1)}
\\
\varphi(0)=0, \varphi\in\mathcal{AC}[0,1]}
}\left\{ \frac12
\int_0^1  \left(\frac{\varphi'(u)}{\sigma(xe^{\varphi(u)})}\right)^2du\right\}.
\end{equation}
\end{proposition}

\begin{proof}
Recall 
\begin{equation*}
{\cal I}_{f}(S_0,\kappa, \tau)=
\inf_{
\substack{\frac{1}{1-\tau}\int_{\tau}^{1}e^{g(t)}dt=\kappa e^{g(1)}
\\
g(0)=\log S_{0}, g\in\mathcal{AC}[0,1]}}\frac{1}{2}\int_{0}^{1}\left(\frac{g'(t)}{\sigma(e^{g(t)})}\right)^{2}dt.
\end{equation*}
If we consider function $f\in\mathcal{AC}[0,1]$ such that $f(t)=g(t)-\log S_0$. Then we can rewrite ${\cal I}_{f}(S_0,\kappa, \tau)$ as
\begin{equation*}
{\cal I}_{f}(S_0,\kappa, \tau)=
\inf_{
\substack{\frac{1}{1-\tau}\int_{\tau}^{1}e^{f(t)}dt=\kappa e^{f(1)}
\\
f(0)=0, f\in\mathcal{AC}[0,1]}}\frac{1}{2}\int_{0}^{1}\left(\frac{f'(t)}{\sigma(S_0e^{f(t)})}\right)^{2}dt.
\end{equation*}
Similarly as in the proof of Proposition \ref{prop1}, by using Lagrange 
multiplier we obtain the optimal path $f(t)$ comes from the family of 
absolutely continuous paths $\{f_c(t)\}_{c\in\R}$ given by
\begin{equation*}
f_c(t)=
\begin{cases}
F^{-1}(ct) & 0<t<\tau\\
F^{-1}(c\tau)+\varphi_c\left(\frac{t-\tau}{1-\tau} \right) & \tau\le t<1
\end{cases}\,,
\end{equation*}
where $\varphi_c(u)$, $0<u<1$ is the argmin of ${\cal I}\left(S_0 e^{F^{-1}(c\tau)}, \kappa \right)$. Here ${\cal I}\left(x, \kappa \right)$ is given as in \eqref{eq-cal-I}
The energy associated to each $f_c(t)$ is given by
\begin{equation*}
 \frac12
\int_0^1  \left(\frac{f_c'(t)}{\sigma(S_0e^{f_c(t)})}\right)^2dt = \frac12 c^2 \tau + \frac{1}{1-\tau}
{\cal I}\left(S_0 e^{F^{-1}(c\tau)}, \kappa \right).
\end{equation*}
Hence we obtain  the rate function
\begin{equation*}
{\cal I}_{f}(S_0,\kappa, \tau)=\inf_{c\in\R}{\cal I}_c=\inf_{c\in\R}\bigg\{ \frac12 c^2 \tau + \frac{1}{1-\tau}
{\cal I}\left(S_0 e^{F^{-1}(c\tau)}, \kappa \right)\bigg\}.
\end{equation*}
and $c$ is chosen to minimize ${\cal I}_c$.
\end{proof}

Under the Black-Scholes model the above result can be simplified significantly.

\begin{proposition}\label{prop-float-s}
Assume the asset price $S_t$ follows the Black-Scholes model, namely $\sigma(\cdot)$ in \eqref{eq-S-t} in a constant $\sigma>0$. Then the rate function of a out-of-the-money floating-strike forward start Asian option is given by
\begin{equation}\label{eq-J-float}
{\cal I}_{f}^{(BS)}(\kappa, \tau,\sigma)=
\frac{\mathcal{J}_{\rm BS}(\kappa)}{(1-\tau)\sigma^2}\,,
\end{equation}
where $\mathcal{J}_{\rm BS}(\cdot)$ is given by \eqref{eq-J-x} and is  
given explicitly in Lemma \ref{lemma-J}.
\end{proposition}

\begin{proof}
From Proposition \ref{prop-f-rate} we know that when $\sigma(\cdot)=\sigma$, the rate function is independent of $S_0$ and given by
\begin{equation*}
{\cal I}_{f}(S_0,\kappa, \tau)= \inf_{c\in \mathbb{R}}
\left\{ \frac12 c^2 \tau + \frac{1}{(1-\tau)\sigma^2}
\inf_{\varphi\in \mathcal{D}(\kappa) }\left\{ \frac12
\int_0^1  \left(\varphi'(u)\right)^2du\right\} \right\},
\end{equation*}
where $\mathcal{D}(\kappa)
=\bigg\{\varphi\in\mathcal{AC}[0,1]\,\bigg|\,\varphi(0)=0, \int_{0}^{1}e^{\varphi(u)}du=\kappa e^{\varphi(1)} \bigg\}$. 
First we claim that
\begin{equation}\label{eq-claim-float-s}
\inf_{\varphi\in \mathcal{D}(\kappa) }\left\{ \frac12
\int_0^1  \left(\varphi'(u)\right)^2du\right\} =\mathcal{J}_{\rm BS}(\kappa)\,.
\end{equation}
The proof can be found in \cite{PZAsian}. 
For completeness we briefly sketch it here. The main idea is to use the 
method of the Lagrange multiplier. Consider the auxiliary variational problem
\begin{equation*}
\Lambda(\varphi):=\frac{1}{2\sigma^2}\int_{0}^{1}\left(\varphi'(t)\right)^{2}dt+\lambda\left(\int_{0}^{1}e^{\varphi(t)}dt-\kappa e^{\varphi(1)} \right)\,,
\end{equation*}
over all functions in $\mathcal{AC}[0,1]$ satisfying $\varphi(0)=0$.
The problem has been reduced to solving the Euler-Lagrange equation
\begin{equation*}
\varphi''(t)=\lambda\sigma^2e^{\varphi(t)}\,,
\end{equation*}
with boundary conditions
\begin{equation*}
\varphi(0)=0,\quad \varphi'(0)=0,\quad  \varphi'(1)=\lambda\sigma^2\kappa e^{\varphi(1)}.
\end{equation*}
Let $h\in \mathcal{AC}[0,1]$ be such that
\begin{equation*}
h(t)=\varphi(1-t)-\varphi(1),
\end{equation*}
hence $\varphi(t)=h(1-t)+\varphi(1)$. The left hand side of \eqref{eq-claim-float-s} is then given by
\begin{equation*}
\inf_{
\substack{\int_{0}^{1}e^{h(t)}dt=\kappa 
\\
h(0)=0,\ h\in\mathcal{AC}[0,1]}
}
\left\{ \frac12
\int_0^1  \left(h'(u)\right)^2du\right\}\,,
\end{equation*}
which is exactly $\mathcal{J}_{\rm BS}(\kappa)$. 
Hence we have \eqref{eq-claim-float-s}, and
\begin{equation*}
{\cal I}^{(BS)}_{f}(\kappa, \tau,\sigma)= \inf_{c\in \mathbb{R}}
\left\{ \frac12 c^2 \tau +\frac{1}{(1-\tau)\sigma^2}\mathcal{J}_{\rm BS}(\kappa)
\right\}
=\frac{1}{(1-\tau)\sigma^2}\mathcal{J}_{\rm BS}(\kappa).
\end{equation*} 
\end{proof}

\begin{remark}
Notice that under Black-Scholes model, the optimization of the rate function 
${\cal I}_{f}(S_0,\kappa, \tau)$ in \eqref{eq-rate-float-s} with respect to 
$c$ is trivial. This comes from the fact that $\I(x,\kappa)$ in 
\eqref{eq-cal-I} is independent of $x$  when $\sigma(\cdot)$ is a constant function. Hence the optimal path remains flat when $0<t<\tau$ in order to achieve smallest energy.
\end{remark}

At last, it is worth investigating an equivalence of floating-strike and 
fixed-strike forward starting Asian options under the Black-Scholes model, 
which was first introduced by Henderson and Wojakowski in \cite{HW} for regular 
Asian options. In the proposition below, we study the equivalence for forward 
start Asian options. Similar results for Asian options with 
discrete time averaging have been discussed in \cite{VDL}, see Theorem 9
in \cite{VDL}.

Denote by $C_f(S_0, \kappa, r, q, T_1, T_2)$ the price of a floating-strike 
call option with floating-strike $\kappa$, forward starting time $T_1$, 
maturity $T_2$, and underlying asset dynamics
\begin{equation}\label{eq-S-BS}
S_t=S_0e^{(r-q-\frac{1}{2}\sigma^2)t +\sigma W_t},\qquad t\geq 0,\quad S_{0}>0,
\end{equation}
where $W_t$ is a standard Brownian motion. 
Denote by $C_x(S_0, K, r, q, T_1, T_2)$ the price of a fixed-strike forward 
start call option with strike price $K$, forward starting time $T_1$, 
maturity $T_2$, and underlying asset following \eqref{eq-S-BS}. Analogously 
denote the prices for put-options: $P_f(S_0, \kappa, r, q, T_1, T_2)$ and 
$P_x(S_0, \kappa, r, q, T_1, T_2)$.
\begin{proposition}
Assume the asset price $S_t$ follows the Black-Scholes model with constant 
volatility $\sigma$. Then we have that
\begin{equation}\label{eq-f-x-1}
C_f(S_0, \kappa, r, q, \tau T, T)=P_x(S_0, \kappa S_0, q, r, 0, (1-\tau)T),
\end{equation}
and
\begin{equation}\label{eq-f-x-2}
P_f(S_0, \kappa, r, q, \tau T, T)=C_x(S_0, \kappa S_0, q, r, 0, (1-\tau)T).
\end{equation}

We also have
\begin{equation}\label{eq-f-x-3}
C_f(S_0, \kappa, r, q, 0, (1-\tau) T)=P_x(S_0, \kappa S_0, q, r, \tau T, T),
\end{equation}
and
\begin{equation}\label{eq-f-x-4}
P_f(S_0, \kappa, r, q, 0, (1-\tau) T)=C_x(S_0, \kappa S_0, q, r, \tau T, T).
\end{equation}

\end{proposition}
\begin{proof}
We prove \eqref{eq-f-x-1}. Since
\begin{align*}
&C_f(S_0, \kappa, r, q, \tau T, T)=e^{-rT}\E\left[ \left(\kappa S_T-\frac{1}{(1-\tau)T}\int_{\tau T}^T S_tdt \right)^+\right]\\
&\qquad=e^{-rT}\E\left[ S_0e^{(r-q)T-\frac12\sigma^2T+\sigma W_T}\left(\kappa-\frac{1}{(1-\tau)T}\int_{\tau T}^T \frac{S_t}{S_T}dt  \right)^+\right]\\
&\qquad =e^{-qT}S_0\,\E\left[e^{-\frac12\sigma^2T+\sigma W_T}\left(\kappa-\frac{1}{(1-\tau)T}\int_{\tau T}^T \frac{S_t}{S_T}dt  \right)^+\right].
\end{align*}
We change the probability measure such that
\[
\frac{dP^*}{dP}=e^{-\frac{1}{2}\sigma^2T+\sigma W_T}.
\]
By Girsanov's theorem we know that $W_t^*=W_t-\sigma t$ is a Brownian motion 
under the probability measure $P$. Hence we have 
\begin{equation}\label{Cfeq1}
C_f(S_0, \kappa, r, q, \tau T, T)=e^{-qT}S_0\,
\E^*\left[\left(\kappa-\frac{1}{(1-\tau)T}\int_{\tau T}^T \frac{S_t}{S_T}dt  
\right)^+\right].
\end{equation}
Moreover, note that
\[
\frac{S_t}{S_T}=e^{(r-q+\frac{1}{2}\sigma^2)(t-T)+\sigma(W_t^*-W_T^*)}\stackrel{\mathcal{D}}{=}
e^{(r-q+\frac{1}{2}\sigma^2)(t-T)+\sigma\hat{W}_{T-t}},
\]
where $\hat{W}$ is a standard Brownian motion under $P^*$. 
By changing the integration variable $t\to T-t$ we obtain that
\begin{equation}\label{Cfeq2}
C_f(S_0, \kappa, r, q, \tau T, T)=e^{-qT}\,\E^*\left[\left(\kappa S_0-\frac{1}{(1-\tau)T}\int_{0}^{(1-\tau) T}S_0e^{(q-r-\frac12\sigma^2)t+\sigma \hat{W}_t}dt  \right)^+\right].
\end{equation}
We then obtain \eqref{eq-f-x-1} by realizing the right hand side of the above 
equation is exactly $P_x(S_0, \kappa S_0,q, r, 0, (1-\tau)T)$.

Similarly, for a floating-strike forward start Asian put option we have
\begin{align*}
&P_f(S_0, \kappa, r, q, \tau T, T)=e^{-rT}\E\left[ \left(\frac{1}{(1-\tau)T}\int_{\tau T}^T S_tdt-\kappa S_T \right)^+\right]\\
&\qquad =e^{-qT}S_0\,\E\left[e^{-\frac12\sigma^2T+\sigma W_T}\left(\frac{1}{(1-\tau)T}\int_{\tau T}^T \frac{S_t}{S_T}dt-\kappa  \right)^+\right]\\
&\qquad\qquad =e^{-qT}\,\E^*\left[\left(\frac{1}{(1-\tau)T}\int_{0}^{(1-\tau) T}S_0e^{(q-r-\frac12\sigma^2)t+\sigma \hat{W}_t}dt-\kappa S_0  \right)^+\right].
\end{align*}
This is indeed \eqref{eq-f-x-2}. The remaining symmetry relations
\eqref{eq-f-x-3}, \eqref{eq-f-x-4} are proved in a similar way.
\end{proof}


\subsection{Generalized Asian options}
\label{sec:4.2}

We consider here a more general type of Asian options with payoff
\begin{eqnarray}
\mbox{Call} &:& (\kappa S_T - A_T + K)^+ \,,
\\
\mbox{Put}  &:& (A_T - \kappa S_T - K)^+\,,
\end{eqnarray}
where the asset price average $A_T$ is defined with respect to the time period
$[0,T]$. Although these options are not forward start, they appear naturally
when pricing so-called seasoned floating-strike 
Asian options. These options are initially forward start, but their
evaluation requires the study of these payoffs when the valuation date
enters the averaging period.
We will call them generalized Asian options, following 
Linetsky \cite{LinetskyRisk}.

The prices of these options are given by expectations in the risk-neutral
measure. For example
\begin{equation}
C_{\rm gen}(S_0,\kappa,K,T )  = e^{-rT} \mathbb{E}\left[\left(\kappa S_T - 
   \frac{1}{T} \int_{0}^T S_t dt + K\right)^+\right]\,.
\end{equation}

The generalized Asian call option is in-the-money if
\begin{equation}
\kappa > 1 - K/S_0\,,
\end{equation}
and at-the-money if $\kappa=1-K/S_0$ and out-of-the-money otherwise.

We study here the short maturity $T\to 0$ asymptotics for the generalized
Asian options in the local volatility model. This is given by a generalization 
of Proposition 21 in \cite{PZAsian}.

\begin{theorem}\label{thm-gen-otm}
Assume that $S_t$ follows the local volatility model \eqref{eq-S-t} and
\eqref{eq-S-t-condition}. 
Then the short maturity asymptotics $T\to0$ for the 
prices of OTM generalized Asian options are as follows.

(1) When $\kappa<1 - K/S_0$, we have
\begin{equation}\label{eq-gen-out-call}
\lim_{T\to0}T\log C_{\rm gen}(S_0,\kappa, K,T)=-{\cal I}_{g}(S_0,\kappa, K),
\end{equation}
where
\begin{equation*}
{\cal I}_{g}(S_0,\kappa, K)=
\inf_{
\substack{\int_{0}^{1}e^{g(t)}dt=\kappa e^{g(1)} + K
\\
g(0)=\log S_{0}, g\in\mathcal{AC}[0,1]}}\frac{1}{2}\int_{0}^{1}\left(\frac{g'(t)}{\sigma(e^{g(t)})}\right)^{2}dt.
\end{equation*}

(2) When $\kappa>1-K/S_0$, we have
\begin{equation}\label{eq-gen-out-put}
\lim_{T\to0}T\log P_{\rm gen}(S_0,\kappa, K,T)=-{\cal I}_{g}(S_0, \kappa, K).
\end{equation}

(3)  When $\kappa=1-K/S_0$, we have 
\begin{equation}\label{eq-atm-gen}
\lim_{T\rightarrow 0}\frac{1}{\sqrt{T}}C_{\rm gen}(\kappa, K,T)
=\lim_{T\rightarrow 0}\frac{1}{\sqrt{T}}P_{\rm gen}(\kappa, K,T)
=\frac{\sigma(S_{0})S_{0}}{\sqrt{2\pi}}\sqrt{\kappa^2 - \kappa + \frac13} \,.
\end{equation}
\end{theorem}

In the Black-Scholes model we can give a more explicit result for the
rate function of the generalized Asian options.

\begin{figure}[ht]
\begin{center}
\includegraphics[height=80mm]{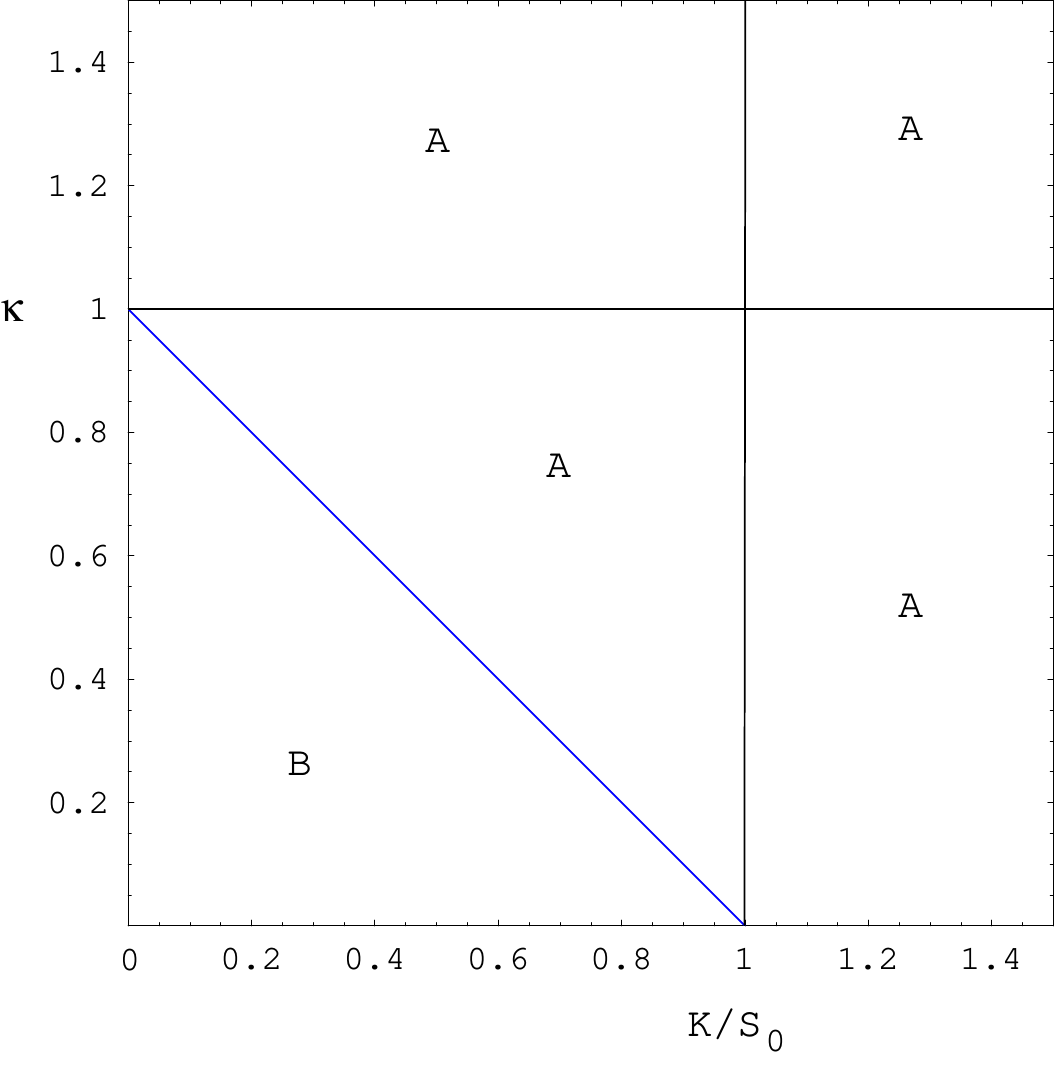}
\end{center}
\caption{Regions in the $(K/S_0,\kappa)$ plane for the rate function of the
generalized Asian option. The points on the blue line correspond to 
ATM generalized Asian options $K/S_0+\kappa=1$. The rate function 
$\mathcal{J}_g(\kappa, K/S_0)$ vanishes along this line.}
\label{Fig:reg}
\end{figure}

\begin{proposition}\label{prop:gen}
The rate function for the generalized Asian option in the Black-Scholes
model has the following form.

(1) In the region A of the $(K/S_0,\kappa)$ plane (above the blue line
$K/S_0+\kappa=1$ in Fig.~\ref{Fig:reg}), we have for the rate function of 
an OTM call 

\begin{eqnarray}\label{J1sol}
\mathcal{J}_g(\kappa,K/S_0) = \sigma^2 \mathcal{I}_g(\kappa,K/S_0)=
\frac12\beta^2 - 2\beta \frac{\gamma}{\gamma+1} \frac{e^\beta-1}{e^\beta+\gamma}\,,
\end{eqnarray}
where 
\begin{eqnarray}\label{gsol}
\gamma= \frac{e^\beta}{\kappa\beta + \sqrt{1+\kappa^2\beta^2}}\,,
\end{eqnarray}
and
$\beta$ is the solution of the equation
\begin{eqnarray}\label{betaeq}
\frac{\gamma+1}{\gamma+e^\beta} \frac{e^\beta-1}{\beta} = 
\kappa e^\beta \left( \frac{\gamma+1}{\gamma + e^\beta}\right)^2 
+ \frac{K}{S_0}\,.
\end{eqnarray}

(2) in the region B of the $(K/S_0,\kappa)$ plane (below the blue line
$K/S_0+\kappa=1$ in Fig.~\ref{Fig:reg}), we have the rate function of 
an OTM put 

\begin{equation}\label{J2sol}
\mathcal{J}_g(\kappa,K/S_0) = \sigma^2 \mathcal{I}_g(\kappa,K/S_0)=
2\xi \left( \tan(\xi+\eta) - \tan\eta - \xi \right)\,,
\end{equation}
where $\eta$ is the solution of the equation 
\begin{equation}\label{etaeq}
\frac{1}{2\xi} \sin (2(\xi+\eta)) = \kappa\,.
\end{equation}
This determines $\eta$ up to a discrete ambiguity
\begin{equation}\label{etasol}
\eta=
\begin{cases}
-\xi + \frac12\arcsin(2\xi\kappa) + n \pi \\
-\xi - \frac12\arcsin(2\xi\kappa) + (n+\frac12) \pi \\
\end{cases}
,\qquad
n \in \mathbb{N}.
\end{equation}
Finally, $\xi$ is given by the solution of the equation
\begin{equation}\label{xieq}
\frac{1}{\xi} \cos^2\eta \left( \tan(\xi+\eta) - \tan\eta \right) =
\kappa \frac{\cos^2\eta}{\cos^2(\xi+\eta)}+ \frac{K}{S_0}\,.
\end{equation}

\end{proposition}

\begin{proof}
The rate function in the Black-Scholes model is given by the solution of the 
variational problem
\begin{equation}
\mathcal{I}_g(\kappa,K) = \mbox{inf}_f \frac{1}{2\sigma^2} \int_0^1 [f'(t)]^2 dt\,,
\end{equation}
where the infimum is taken over all functions $f(t)$ satisfying $f(0)=0$
and 
\begin{equation}\label{intconst0}
\int_0^1 e^{f(t)} dt = \kappa e^{f(1)} + \frac{K}{S_0}\,.
\end{equation}
The proof follows closely that of Prop.~23 in \cite{PZAsian} and will be
omitted. 
\end{proof}

We note a few properties of the rate function for generalized Asian options.

\begin{proposition}
The rate function $\mathcal{J}_g(\kappa,K/S_0)$ for the generalized
Asian options in the Black-Scholes model has the following properties.

(i) The rate function vanishes along the ATM line $\kappa+K/S_0=1$.

(ii) It is symmetric under the exchange of its arguments
\begin{equation}
\mathcal{J}_g(\kappa,K/S_0) = \mathcal{J}_g(K/S_0,\kappa)\,.
\end{equation}

\end{proposition}

\begin{proof}

(i) It is easy to see that $f(t)=0$ satisfies
the integral constraint (\ref{intconst0}) provided that 
$\kappa + \frac{K}{S_0}=1$. Thus the optimal path along this line is $f(t)=0$ 
and the rate function vanishes. In fact it is easy to see that this holds
also in the more general local volatility model and is not specific to the
Black-Scholes model.

(ii) The integral constraint (\ref{intconst0}) for $\mathcal{J}_g(\kappa,K/S_0)$
can be written equivalently in terms of the function $h(t)$
defined by $f(t) = h(1-t) + f(1)$ as
\begin{equation}
e^{f(1)} \int_0^1 e^{h(t)} dt = \kappa e^{f(1)} + \frac{K}{S_0}\,.
\end{equation}
Noting that $h(0)=0,h(1) = - f(1)$, this can be expressed further as
\begin{equation}
\int_0^1 e^{h(t)} dt = \kappa  + \frac{K}{S_0} e^{h(1)}\,,
\end{equation}
which is just the integral constraint for $\mathcal{J}_g(K/S_0,\kappa)$. 
The rate function has the same form when expressed in terms of $h(t)$
by noting that $f'(t) = - h'(1-t)$. Thus the optimal path for 
$\mathcal{J}_f(\kappa,K/S_0)$
is mapped by this change of variable into the optimal path for 
$\mathcal{J}_g(K/S_0,\kappa)$, and the rate function takes the same value for
both cases.
This concludes the proof of the result.

\end{proof}


\subsection{Numerical tests for forward start floating strike Asian options}
\label{sec:4.3}

We consider here the pricing of floating strike forward start Asian options.
Following the same approach as in \cite{PZAsianCEV}, we consider them as
call/put options on the underlying $B_{[\tau T,T]} := \kappa S_T - 
A_{[\tau T,T]}$. For $\kappa \geq 0$ the support of $B_{[\tau T,T]}$
is the entire real axis. A normal approximation is thus more appropriate
than a log-normal approximation for the distribution of this random variable.

The forward price of the $B_{[\tau T,T]}$ asset is
\begin{equation}
F_f(\tau T,T) := \mathbb{E}[B_{[\tau T,T]}] = S_0 \left( \kappa e^{(r-q)T} -
\frac{e^{(r-q)T} - e^{(r-q)\tau T}}{(r-q)(1-\tau)T} \right)\,.
\end{equation}

We propose the price approximation expressed as the Bachelier formula
\begin{equation}\label{Csol}
C_f(\kappa,\tau) = e^{-r T} \left(F_f(\tau T,T) N(d) - \frac{1}{\sqrt{2\pi}}
\Sigma_N \sqrt{T} e^{-\frac12 d^2} \right)\,,
\end{equation}
with $d=\frac{F_f(\tau T,T)}{\Sigma_N\sqrt{T}}$. This corresponds to a 
call option on the underlying $B_{[\tau T,T]}$ with zero strike, which is
the payoff  of the option considered $(B_{[\tau T,T]})^+$.

Requiring agreement with the short maturity option pricing asymptotics gives the 
following short maturity asymptotics for the normal equivalent volatility
$\Sigma_N(K,T)$.

\begin{proposition}
Assume $r=q=0$. 

(1) the short-maturity limit $T\to 0$ of the normal (Bachelier) equivalent
volatility of an OTM forward start floating strike Asian option in the BS
model is
\begin{equation}
\lim_{T\to 0} \Sigma_{\rm N}^2(\kappa,T) = \sigma^2 \frac{(K-S_0)^2}
{2\mathcal{J}_{\rm BS}(\kappa)} (1-\tau)\,.
\end{equation}

(2) the short-maturity limit $T\to 0$ of the normal (Bachelier) equivalent
volatility of an ATM $\kappa=1$ forward start floating strike Asian option 
in the BS model is
\begin{equation}\label{SigNATM}
\lim_{T\to 0} \Sigma_{\rm N}(1,T) = \sigma S_0 \frac{1}{\sqrt{3}} \sqrt{1-\tau}\,.
\end{equation}
\end{proposition}

\begin{proof}
The proof is analogous to the proof of Prop.~17 in \cite{PZAsianCEV} and
will be omitted.
\end{proof}

We consider three benchmark cases, corresponding to a
forward start floating strike call
options with $\kappa=1$ with payoff $\left(S_T - A_{[\tau T,T]}\right)^+$ and
maturity $T=1$ year. The averaging period is as follows
\begin{eqnarray}
\mbox{Case I:} && (1-\tau)T = 60 \mbox{ days}\,,\quad 
\tau = 1- \frac{60}{365}= \frac{305}{365}\,; \\
\mbox{Case II:} && (1-\tau)T = 100 \mbox{ days}\,,\quad 
\tau = 1- \frac{100}{365}= \frac{265}{365}\,; \\
\mbox{Case III:} && (1-\tau)T = 182 \mbox{ days}\,,\quad 
\tau = 1- \frac{182}{365}= \frac{183}{365}\,. 
\end{eqnarray}
They correspond to the benchmark case considered
in Tables II,III and IV of \cite{TCL}. The initial asset price
is $S_0=100$, the interest rate is $r=10\%$ and dividend rate $q=0$.

The forward $F_f(\tau T,T)$ for the three cases is 
$F_f(\tau T,T) = 0.9034\,, 1.5\,, 2.71$,
which is positive. Thus the options are in-the-money. However in the
short maturity limit $T\to 0$ these options would be ATM since $\kappa=1$.
For this reason we will use the ATM normal implied volatility (\ref{SigNATM})
in the numerical estimates below.

In Table~\ref{tab:3} we show results for the benchmark cases considered
in Tables II,III,IV of \cite{TCL}. The results of the
analytical approximation of \cite{TCL}, and the result of MC simulation,
are compared against the short maturity asymptotic result following from
Eq.~(\ref{Csol}). The asymptotic result is lower than the MC simulation
result. The asymptotic result is sensitive to the interest rate $r$ through
the forward rate $F_f(\tau T,T)$. The deviation from the exact result increases
with $r$, as seen in the numerical tests in \cite{PZAsian}. For the case
considered here the interest rate $r=10\%$ is rather large, which explains 
the larger
difference comparing with the benchmark results. However we note qualitative
agreement with the benchmark scenarios, as the differences with the MC
simulation is always below 10\% in relative value.

\begin{table}[!ht]
\caption{\label{tab:3} 
Numerical tests for the forward start floating strike Asian options,
comparing with benchmark test results in \cite{TCL}. $C_f(1,\tau)$ gives the 
result of the short maturity asymptotic approximation (\ref{Csol}).
TCL shows the analytical approximation of \cite{TCL} and MC shows the result
of a Monte Carlo simulation.}
\begin{center}
\begin{tabular}{ccccc}
\hline
$\sigma$ & $\tau$ & $C_f(1,\tau)$ & TCL & MC  \\
\hline\hline
0.2 & 305/365 & 2.13005  & 2.304166 & 2.264616 \\
0.3 & 305/365 & 2.96462  & 3.223310 & 3.171457 \\
0.4 & 305/365 & 3.80438  & 4.142445 & 4.084590 \\
0.5 & 305/365 & 4.64622  & 5.058229 & 4.999982 \\
0.6 & 305/365 & 5.48911  & 5.969642 & 5.908669 \\
\hline
0.2 & 265/365 & 2.92733  & 3.147266 & 3.105413 \\
0.3 & 265/365 & 3.99604  & 4.320233 & 4.268604 \\
0.4 & 265/365 & 5.07577  & 5.493439 & 5.493439 \\
0.5 & 265/365 & 6.15994  & 6.660697 & 6.620761 \\
0.6 & 265/365 & 7.24634  & 7.821191 & 7.790436 \\
\hline
0.2 & 183/365 & 4.33054  & 4.609475 & 4.554739 \\
0.3 & 183/365 & 5.74906  & 6.152161 & 6.096058 \\
0.4 & 183/365 & 7.19388  & 7.696422 & 7.661532 \\
0.5 & 183/365 & 8.64938  & 9.231254 & 9.228435 \\
0.6 & 183/365 & 10.1102  & 10.76042 & 10.798441 \\
\hline
\end{tabular}
\end{center}
\end{table}

\appendix
\section{Notations}

We summarize here the notations for the rate functions used in the main text.

\begin{itemize}

\item $\mathcal{I}(S_0,K,\tau)$ is the rate function for an
Asian option with averaging starting at time zero in the local volatility model.

\item $\mathcal{I}_{\rm fwd}(S_0,K,\tau)$ is the rate function for a forward 
start Asian option in the local volatility model. This depends on $S_0,K$ in a 
non-trivial way.

\item $\mathcal{I}_{\rm fwd}^{(BS)}(K/S_0,\sigma,\tau) = 
\frac{1}{\sigma^2} \mathcal{J}_{\rm fwd}^{(BS)}(K/S_0,\tau)$ is the rate 
function for
a forward start Asian option in the Black-Scholes model. This depends only
on the ratio $K/S_0$. 

\item $\mathcal{I}_{\rm BS}(K/S_0,\sigma) = 
\frac{1}{\sigma^2} \mathcal{J}_{\rm BS}(K/S_0)$ is the rate function for
an Asian option in the Black-Scholes model with averaging starting at time 0. 
This corresponds to $\tau=0$ and is related to the forward start rate
functions defined above as $\mathcal{I}_{\rm BS}(K/S_0,\sigma) = 
\mathcal{I}_{\rm fwd}^{(BS)}(K/S_0,\sigma,0)$ and
$\mathcal{J}_{\rm BS}(K/S_0)=\mathcal{J}_{\rm fwd}^{(BS)}(K/S_0,0)$.

\item $\mathcal{I}_f(S_0,\kappa,\tau)$
is the rate function for a forward start floating-strike Asian option 
in the local volatility model. 

\item $\mathcal{I}_f^{(BS)}(\kappa,\tau,\sigma)$
is the rate function for a forward start floating-strike Asian option 
in the Black-Scholes model. 

\item ${\cal I}_{g}(S_0,\kappa, K)$ is the rate function for the
generalized Asian option.

\item ${\cal J}_{g}(\kappa, K/S_0)$ is the rate function for the
generalized Asian option in the Black-Scholes model.

\end{itemize}

\section*{Acknowledgements}
Lingjiong Zhu is partially supported by NSF Grant DMS-1613164.


\begin{thebibliography}{100}


\bibitem{MLP}
Arguin, L.P., N.-L.~Liu and T.-H.~Wang. (2017)
Most-likely-path in Asian option pricing under local volatility models,
\textit{arXiv:1706.02408}.

\bibitem{ABBF} 
Avellaneda, M., D.~Boyer-Olson, J.~Busca and P.~Friz. (2002).
Reconstructing volatility. 
\textit{Risk} 91, Oct.~2002.

\bibitem{BBF}
Berestycki, H., Busca, J. and I. Florent. (2002).
Asymptotics and calibration of local volatility models. 
\textit{Quantitative Finance}. 
\textbf{2}, 61-69.

\bibitem{BaldiCaram}
Baldi, P. and L.~Caramellino. (2011).
General Freidlin-Wentzell large deviations and positive diffusions. 
\textit{Stat.~Prob.~Letters}. 
\textbf{81}, 1218-1229.

\bibitem{Bergomi}
Bergomi, L. (2004)
Smile Dynamics I.
\textit{Risk}, September, 2004; 
Smile Dynamics II, 
\textit{Risk}, March 2005;
Smile Dynamics III, 
\textit{Risk}, March 2008;
Smile Dynamics IV, 
\textit{Risk}, December 2009.

\bibitem{Bliss}
Bliss, G.~A. (1946)
Lectures on the calculus of variations.
The University of Chicago Press.

\bibitem{BBC}
Bouaziz, L., E. Briys and M.~Crouhy (1994).
The pricing of forward-starting Asian options.
\textit{Journal of Banking and Finance}.
\textbf{18}, 823-839.

\bibitem{BoylePot}
Boyle, P. and Potapchik (2008).
Prices and sensitivities of Asian options: A survey.
\textit{Insurance: Mathematics and Economics}.
\textbf{42}, 189-211.

\bibitem{Buehler}
Buehler, H. (2006).
Consistent variance curve models.
\textit{Finance and Stochastics}  
\textbf{10}, 178-203.

\bibitem{quanto}
Chang, Chuang-Chang, Tzu-Hsiang Liao and Chueh-Yung Tsao (2011).
Pricing and hedging quanto forward-starting floating-strike Asian options.
\textit{The Journal of Derivatives}
\textbf{18}, 37-53.

\bibitem{CH}
Courant, R. and D.~Hilbert. (1953).
Methods of Mathematical Physics, vol.1
Interscience Publishers, New York.

\bibitem{Curran}
Curran, M. (1994).
Valuing Asian options and portfolio options by conditioning on the geometric
mean price.
\textit{Management Science}.
\textbf{40}, 1705-1711.

\bibitem{DZ} 
Dembo, A. and O. Zeitouni. (1998).
\textit{Large Deviations Techniques and Applications}. 
2nd Edition, Springer, New York, 1998.

\bibitem{DGY}
Donati-Martin, C., Ghomrasni, R. and M. Yor. (2001).
On certain Markov processes attached to exponential functionals of Brownian
motion; application to Asian options.
\textit{Rev.~Math.~Iberoam.} {\bf 17}, 179-193.

\bibitem{DRYZ}
Donati-Martin, C., Rouault, R., M. Yor and M.~Zani. (2004).
Large deviations for squares of Bessel and Ornstein-Uhlenbeck processes.
\textit{Probab.~Theory Related Fields} {\bf 129}, 261-289.

\bibitem{Du}
Dufresne, D. (2000).
Laguerre series for Asian and other options.
\textit{Math.~Finance}
\textbf{10}, 407-428.

\bibitem{DufresneLN} 
Dufresne, D. (2004).
The lognormal approximation in financial and other computations. 
\textit{Adv.~Appl.~Prob}. 
\textbf{36}, 747-773.

\bibitem{DufresneReview} 
Dufresne, D. (2005).
Bessel processes and a functional
of Brownian motion, in M.~Michele and H.~Ben-Ameur (Ed.), 
\textit{Numerical Methods in Finance}, 35-57, Springer, 2005.

\bibitem{EM}
Eberlein, E. and D. B. Madan. (2009).
Sato processes and valuation of structured products.
\textit{Quantitative Finance}.
\textbf{9}, 27-42.

\bibitem{FPP2013}
Foschi, P., Pagliarani, S. and A. Pascucci. (2013).
Approximations for Asian options in local volatility models.
\textit{Journal of Computational and Applied Mathematics}.
\textbf{237}, 442-459.

\bibitem{FMW}
Fu, M., Madan, D. and T. Wang. (1998).
Pricing continuous time Asian options: a comparison of Monte Carlo and Laplace
transform inversion methods.
\textit{J.~Comput.~Finance}
\textbf{2}, 49-74.

\bibitem{Gatheral}
Gatheral, J. (2006).
\textit{The Volatility Surface: A Practitioner's Guide}.
John Wiley \& Sons, 2006.

\bibitem{GHLOW}
Gatheral, J., E.P.~Hsu, P.~Laurent, C.~Ouyang and T.-H. Wang. (2012).
Asymptotics of implied volatility in local volatility models.
\textit{Math.~Fin.}
\textbf{22}, 591-620.

\bibitem{GW}
Gatheral, J. and T.-H. Wang. (2012).
The heat-kernel most-likely-path approximation.
\textit{IJTAF}.
\textbf{15}, 1250001.

\bibitem{GY}
Geman, H. and M. Yor. (1993).
Bessel processes, Asian options and perpetuities.
\textit{Math.~Finance}
\textbf{3}, 349-375.

\bibitem{GlassermanWu}
Glasserman, P. and Q. Wu. (2011).
Forward and future implied volatility.
\textit{IJTAF}.
\textbf{14}, 407-432.

\bibitem{HW}
Henderson, V. and R. Wojakowski. (2002).
On the Equivalence of Floating and Fixed-Strike Asian Options.
\textit{International Journal of Theoretical and Applied Finance}.
\textbf{39}, 391-394.

\bibitem{HLbook}
Henry-Labord\'ere, P. (2008).
Analysis, Geometry and Modeling in Finance.
Chapman and Hall, CRC Financial Mathematics Series.

\bibitem{JackRoome}
Jacquier, A. and P.~Roome. (2015).
Asymptotics of forward implied volatility. 
\textit{SIAM J. Finan. Math.} 
\textbf{6}, 307-351. 

\bibitem{JackRoome2}
Jacquier, A. and P.~Roome. (2013).
The small-maturity Heston forward smile.
\textit{SIAM J. Finan. Math.}
\textbf{4}, 831-856.

\bibitem{Levy}
Levy, E.(1992).
Pricing European average rate currency options.
\textit{Journal of International Money and Finance}
\textbf{11}, 474-491.

\bibitem{LinetskyRisk}
Linetsky, V. (2002). 
Exotic spectra.
\textit{Risk}
April 2002.

\bibitem{Linetsky}
Linetsky, V. (2004). 
Spectral expansions for Asian (Average price) options.
\textit{Operations Research}
\textbf{52}, 856-867.

\bibitem{MR}
Musiela, M. and M. Rutkowski. (2005).
\textit{Martingale Methods in Financial Modelling}.
2nd Edition.
Springer, 2005.

\bibitem{PZAsian} 
Pirjol, D. and L. Zhu. (2016).
Short maturity Asian options in local volatility models. 
\textit{SIAM Journal on Financial Mathematics}.
\textbf{7}, 947-992.

\bibitem{PZAsianCEV} 
Pirjol, D. and L. Zhu. (2017). 
Short maturity Asian options for the CEV model, 
\textit{arXiv:1702.03382}.

\bibitem{PZdiscrete}
Pirjol, D., L.~Zhu. (2017).
Asymptotics for the discrete time average
of the geometric Brownian motion and Asian options, 
\textit{Adv.~Appl.~Prob.}
\textbf{49}, 446-480.

\bibitem{RogersShi}
Rogers, L. and Z. Shi. (1995).
The value of an Asian option.
\textit{J.~Appl.~Prob.}
\textbf{32}, 1077-1088.

\bibitem{RogersTeh}
Rogers, L. and M. R. Tehranchi. (2010).
Can the implied volatility surface move by parallel shifts?
\textit{Finance and Stochastics}
\textbf{14}, 235-248.

\bibitem{Schon}
Sch\"onbucher, P.J. (1999).
A market model for stochastic implied volatility.
\textit{Phil.~Trans.~Royal Soc. London}
\textbf{A357}, 2071-2092.

\bibitem{TavellaRandall}
Tavella, D. and C.~Randall. (2000).
Pricing financial instruments - the finite difference method.
Wiley, 2000.

\bibitem{TCL}
Tsao, Chueh-Yung, Chuang-Chang Chang and Chung-Gee Lin. (2003)
Analytic approximation formulae for pricing forward-starting Asian options.
\textit{The Journal of Futures Markets}
\textbf{23}, 487-516.

\bibitem{VDL}
Vanmaele, M., Deelstra, G., Liinev, J., Dhaene, J. and Goovaerts, M.J. (2006).
Bounds for the price of discrete arithmetic Asian options,
\textit{J.~Comp.~Appl.~Math.}
\textbf{185}, 51-90.

\bibitem{Varadhan67}
Varadhan, S. R. S. (1967).
Diffusion processes in a small time interval.
\textit{Communications on Pure and Applied Mathematics}.
\textbf{20}, 659-685.

\bibitem{Vecer}
Vecer, J. (2001).
A new PDE approach for pricing arithmetic average Asian options.
\textit{J.~Comp.~Finance}
\textbf{4}(4), 105-113.

\bibitem{VecerRisk}
J.~Vecer. (2002).
Unified Asian pricing.
\textit{Risk}.
\textbf{15}, 113-116.

\bibitem{thumb}
Personal communication to DP from Dyutiman Das.


\end{thebibliography}
\end{document}